\documentclass[11pt,a4paper]{amsart}
\usepackage[T1]{fontenc}
\usepackage[utf8]{inputenc}
\usepackage{lmodern,microtype}
\usepackage[margin=21mm]{geometry}
\usepackage{setspace}
\usepackage{amsmath,amssymb,amsthm,mathtools}
\usepackage{booktabs,array,enumitem}
\usepackage[numbers,sort&compress]{natbib}
\usepackage{xurl}
\usepackage[unicode,hidelinks]{hyperref}
\hypersetup{pdftitle={Dictators are most informative},
            pdfauthor={Vu Khac Ky and Tuan Tran}}
\newcommand{\E}{\mathbb E}
\newcommand{\R}{\mathbb R}
\newcommand{\PP}{\mathbf P}

\DeclareMathOperator{\Cov}{Cov}
\DeclareMathOperator{\Ent}{Ent}
\DeclareMathOperator{\Var}{Var}
\DeclareMathOperator{\atanh}{atanh}
\newtheorem{theorem}{Theorem}[section]
\newtheorem{conjecture}[theorem]{Conjecture}
\newtheorem{lemma}[theorem]{Lemma}
\newtheorem{proposition}[theorem]{Proposition}
\newtheorem{corollary}[theorem]{Corollary}
\numberwithin{equation}{section}
\allowdisplaybreaks[2]
\setlist{itemsep=3pt,topsep=5pt}

\title[Dictators are most informative]
{Dictators are most informative}

\author{Vu Khac Ky}
\address{Department of Mathematics, FPT University, Hanoi, Vietnam}
\email{kyvk2@fpt.edu.vn}

\author{Tuan Tran}
\address{School of Mathematical Sciences, University of Science and Technology of China, Anhui, China}
\thanks{Tuan Tran was supported by the Excellent
Young Talents Program (Overseas) of the National Natural Science Foundation of China under Grant No.
GG0010007003.}
\email{trantuan@ustc.edu.cn}

\begin{document}
\flushbottom

\begin{abstract}
We prove the Courtade--Kumar conjecture: among all Boolean functions
$f\colon \{-1,1\}^n\to\{-1,1\}$, a dictator retains the most
information about a uniformly random input observed through independent
binary noise.
\end{abstract}

\maketitle
\pagestyle{plain} 

\section{Introduction}\label{sec:introduction}

How much information about a noisy vector can be retained in one bit?
Let $X$ be uniform on $\{-1,1\}^n$ and, for $0\le\rho\le1$, obtain
$Y_\rho$ from $X$ by independently flipping each coordinate $X_i$ with probability
$(1-\rho)/2$, so that larger $\rho$ means less noise. A one-bit
summary of $X$ has the form $f(X)$ for a Boolean function
$f:\{-1,1\}^n\to\{-1,1\}$, and we seek a summary that shares as
much information as possible with $Y_\rho$. Courtade and Kumar
\cite{CourtadeKumar2014} relate this question to the information
bottleneck problem, where a compressed description should retain
information about a correlated observation, and to the value of one
bit of side information.

To measure the information preserved by a summary, recall that the
\emph{Shannon entropy} of a finite random variable $U$ with probability
mass function $p_U$ is $H(U)=-\sum_u p_U(u)\log p_U(u)$.
We use natural logarithms and the convention $0\log0=0$.
The \emph{conditional entropy}
$H(U\mid V)=\sum_{v:p_V(v)>0}p_V(v)H(U\mid V=v)$ averages the
entropy of the conditional distributions of $U$ given $V$.
The \emph{mutual information} $I(U;V)=H(U)-H(U\mid V)$ is
therefore the average reduction in uncertainty about $U$ obtained
by observing $V$. Our objective is to maximize $I(f(X);Y_\rho)$
over all Boolean functions $f$.

Kumar and Courtade conjectured in 2013 \cite{KumarCourtade2013}
that a single input coordinate is optimal, and studied the conjecture
further in \cite{CourtadeKumar2014}. Such a summary, allowing a change
of sign, is a \emph{signed dictator}: $f(x)=\sigma x_i$, where
$1\le i\le n$ and $\sigma\in\{-1,1\}$, with random output $\sigma X_i$.
To express its information value, write
$H_{\rm b}(p)=-p\log p-(1-p)\log(1-p)$ for binary entropy, and set
\begin{equation}\label{eq:entropy-functions}
 H(t)=H_{\rm b}((1-t)/2),\qquad \Phi(t)=\log2-H(t).
\end{equation}
Here $H(t)$ is the entropy of a sign-valued variable with mean $t$.
A Boolean function is
called \emph{balanced} when $\E f=0$. Every signed dictator is
balanced, with entropy $\log2$, and its conditional entropy given
$Y_\rho$ is $H(\rho)$. Its mutual information is therefore
$\Phi(\rho)$, which gives the benchmark in the following conjecture.

\begin{conjecture}[Courtade--Kumar {\cite{KumarCourtade2013,CourtadeKumar2014}}]\label{conj:ck}
For every integer $n\ge1$, every Boolean function
$f:\{-1,1\}^n\to\{-1,1\}$, and every $\rho\in[0,1]$,
\[
 I(f(X);Y_\rho)\le\Phi(\rho).
\]

\end{conjecture}

The problem became a meeting point of information theory and Boolean
analysis. Kindler, O'Donnell and Witmer
\cite{KindlerODonnellWitmer2015} record the attention it attracted,
while Yu and Tan's 2022 monograph calls it “one of the most important
open problems in information theory” \cite[Section~9.1.2]{YuTan2022}.
The latter places it alongside noise stability, noninteractive
correlation distillation, and hypercontractivity.

The conjecture has led to sustained work on several related extremal
problems. Anantharam, Gohari, Kamath and Nair
\cite{AnantharamGohariKamathNair2013} used hypercontractivity to study
the weaker problem in which both vectors are reduced to Boolean
outputs; Pichler, Piantanida and Matz \cite{PichlerPiantanidaMatz2018}
later proved this two-output inequality. Kindler, O'Donnell and Witmer \cite{KindlerODonnellWitmer2015}
examined fixed-mean and continuous formulations, while Ordentlich,
Shayevitz and Weinstein \cite{Ordentlich2016} improved the general
upper bound by Fourier analysis. Samorodnitsky
\cite{Samorodnitsky2016} established CK for
$0\le\rho\le\rho_0$, with an absolute constant $\rho_0>0$. Yu \cite{Yu2023,Yu2026Local}
subsequently strengthened this result in the balanced case
to the explicit range $0\le\rho\le0.914$, while also obtaining
further bounds for arbitrary means.

A second line of work seeks a stronger functional inequality from
which CK would follow. Anantharam, Bogdanov, Chakrabarti, Jayram and
Nair \cite{AnantharamEtAl2017} proposed a Hellinger-entropy strengthening
and proved that it implies CK. Chen and Nair \cite{ChenNair2024}
placed these conjectures in an ordered family of $\Phi$-entropy
inequalities and related their small-noise limits to isoperimetry.
Their introduction also records extensive work on the problem at a
Simons workshop in 2015. Chen, Gohari and Nair
\cite{ChenGohariNair2025} then developed a differential-equation
approach and reduced a sufficient condition for balanced CK to
four explicit inequalities. This program connects the conjecture
to both functional inequalities and the evolution of entropy under
noise. The ordering in \cite{ChenNair2024} concerns the entropy
parameter; it does not establish downward propagation in $\rho$.

Li and M\'edard \cite{LiMedard2021} studied noisy moments and
noninteractive correlation distillation. Barnes and \"Ozg\"ur
\cite{BarnesOzgur2020} proved an equivalence between balanced CK and
a symmetrized Li--M\'edard conjecture. More recent work includes
the coordinate-wise information bound of Javanmard and Woodruff
\cite{JavanmardWoodruff2026}; its objective differs from information
about the full noisy vector.

In this paper, we establish the Courtade--Kumar conjecture for every mean and correlation.  Our proof uses log-Sobolev contraction
\cite{Gross1975}, Harris association \cite{Harris1960}, and martingale
entropy methods \cite{FalikSamorodnitsky2007}, with the Fourier
normalization of \citet{ODonnell2014}.

\begin{theorem}\label{jp:main}
The Courtade--Kumar conjecture holds.
\end{theorem}

The next result, proved in the companion paper \cite{VuTranStability},
is a stability version of Theorem~\ref{jp:main}:
if $I(f(X);Y_\rho)$ is close to $\Phi(\rho)$, then $f$ is close
to a signed dictator.

\begin{theorem}[Stability]\label{st:global}
Let $0<\rho<1$ and $\varepsilon\ge0$. If
$f:\{-1,1\}^n\to\{-1,1\}$ satisfies
\[
I(f(X);Y_\rho)\ge\Phi(\rho)-\varepsilon,
\]
then there are $1\le i\le n$ and $\sigma\in\{-1,1\}$ such that
\[
\PP\big(f(X)\ne\sigma X_i\big)
 \le\frac{10^8\,\varepsilon}{\rho^2(1-\rho)H(\rho)}.
\]
\end{theorem}

\medskip
\noindent \textbf{Independent work.}
While preparing the present manuscript, we learned that Zijie Chen, Amin Gohari, Adel Javanmard, Honghao Lin, Vahab Mirrokni, Chandra Nair, and David P. Woodruff \cite{CGJLMNW2026} had independently obtained a proof of the Courtade--Kumar conjecture. Their proof is substantially different from ours: their argument develops the differential-equation approach, whereas ours develops a new entropy-production and spectral framework. We thank them for their collegiality in coordinating the simultaneous posting of the two manuscripts on arXiv.

\subsection{Methodology and relation to previous work}

Our main technical contribution is a set of entropy-production
and spectral estimates that control functions whose dependence
is spread across many coordinates. These estimates retain both
the contributions of selected coordinates and the variation
outside them. Combined with local entropy comparisons near
dictators, they give a bound valid across all noise levels,
uniformly in the dimension. The key step is to make these two
types of comparison meet without losing the information needed
near the noiseless endpoint.

The difficulty at that endpoint is intrinsic to the problem.
When $\rho=1$, we have $Y_1=X$, so every balanced Boolean function
carries the same information $\log2$. For example, parity on
$k\ge2$ coordinates carries $\Phi(\rho^k)$ information, which
approaches the dictator value $\Phi(\rho)$ as $\rho$ tends to one,
although its probability distance from every signed dictator
remains $1/2$. A qualitative proof must therefore resolve the
sign of a vanishing deficit even for functions far from
dictators. At the same time, higher Fourier levels become less
strongly suppressed as the noise decreases, making estimates
that discard their contribution insufficient.

Several ingredients come from earlier work. Samorodnitsky
\cite{Samorodnitsky2016} proved CK at sufficiently high noise,
and Yu \cite{Yu2023,Yu2026Local} substantially enlarged the
known range for balanced functions. Yu's Fourier bounds and
local entropy comparisons provide important inputs to our
argument. At low noise, Courtade and Kumar
\cite[Appendix~B, Theorem~6]{CourtadeKumar2014} used edge
isoperimetry to obtain dictator optimality in each fixed
dimension. These endpoint results leave a dimension-independent
comparison across all noise levels unresolved. We derive the
needed forms of the local inputs to retain the dependence on
the mean and to connect them to our spectral estimates.

Our estimates address this gap by comparing the rate at which
noise creates conditional entropy with quadratic energy.
Keeping track of where this energy lies allows us to distinguish
dependence concentrated on one coordinate from dependence spread
among several. We use the comparison in two complementary ways:
integration gives an entropy bound in the intermediate range,
while a differential argument rules out an excess over the
dictator benchmark near the noiseless endpoint. Throughout,
we retain the smaller initial entropy of an unbalanced bit,
so the argument also covers arbitrary means.

Concavity and convexity reductions then turn the remaining
comparisons into scalar inequalities independent of the
dimension. This separates the analytic mechanism from the
finite verification needed in the intermediate ranges; the
final approach to the noiseless endpoint is analytic.
Retaining quantitative margins in these comparisons leads
to the stability theorem proved in the companion paper
\cite{VuTranStability}.

\subsection{Outline of the proof}
The proof compares the uncertainty left after observing the noisy
vector with the uncertainty left by a dictator. For a fixed
$f:\{-1,1\}^n\to\{-1,1\}$ with mean $\mu=\E f$, write
\begin{equation}\label{eq:information-entropy}
 A_\rho(f)=I(f(X);Y_\rho),\qquad E_\rho(f)=H(\mu)-A_\rho(f).
\end{equation}
Thus $E_\rho(f)$ is the conditional
entropy of the bit $f(X)$ given $Y_\rho$. We must show that the gap
\begin{equation}\label{ub:gap-definitions}
 G_f(\rho)=A_\rho(f)-\Phi(\rho)
          =H(\rho)-E_\rho(f)-\Phi(\mu)
\end{equation}
is nonpositive. Keeping the term $\Phi(\mu)$ is essential: the
initial entropy of an unbalanced bit is smaller than that of a dictator.

A compression argument lets us assume that $f$ is increasing in each
coordinate. We then distinguish functions close to a single coordinate
from functions whose dependence is spread across several coordinates.
For the first group, conditioning on one or two coordinates gives a
direct entropy bound. For the second group, Fourier analysis measures
how much dependence remains outside those coordinates.

To control the sources whose dependence is spread across coordinates,
we study the rate of entropy change. Let $T_\rho f(y)=\E[f(X)\mid Y_\rho=y]$ be the noise operator.
Then $E_\rho(f)=\E H(T_\rho f)$. Differentiating this expression
reduces the problem to an inequality on each edge of the cube.
The edge inequality gives a lower bound for the rate at which noise
creates entropy, in terms of the current entropy and the coordinate
influences of $f$. Integrating that bound from the noiseless input gives
a lower bound for $E_\rho(f)$. Section~\ref{sec:profile-clocks}
writes this integration as a change of variables; the resulting
integral is called a \emph{profile clock} there. The middle range is covered by the local bounds, this integrated
entropy bound, and a direct Fourier-energy estimate.

At correlations close to one, the same comparison is used
differentially. A stronger edge inequality, combined with the spectral comparison
behind Theorem~\ref{fh:entropy-smoothing}, controls $G_f'(\rho)$.
If $G_f(\rho)>0$, the bound makes this derivative nonnegative.
Such a positive gap could not return to a nonpositive value at
$\rho=1$, where $G_f(1)=-\Phi(\mu)\le0$. The spectral estimate
reaches the neighborhood already covered by the local entropy bound,
so these two estimates account for every source. The source is
held fixed throughout this differential argument.

The same scalar entropy profile supports both arguments.
The decrease and log-concavity of $vH(\tanh v)/\tanh v$
(Lemma~\ref{as:profile-convex}) let us replace a long list of
influences by a few extremal values. They also allow an integrated
entropy bound checked at one correlation to cover every larger
correlation. A separate local bound propagates toward smaller
correlations. These are monotonicity statements for the bounds used
in the proof; they do not assert monotonicity of the conjecture itself.

These analytic reductions eliminate the dimension and the full list of influences, leaving explicit scalar inequalities. 
In the central range, affine interpolation reduces the scalar entropy
bound to two noise endpoints, and a series argument controls the
energy remainder throughout the interval. In the middle range, the integral is
convex in the retained squared influences on each packing region
(Lemmas~\ref{lem:source-convexity} and~\ref{lem:moving-cap}),
so its maximum occurs at a boundary
value. Supporting lines then make each integral explicit, without
certifying roots of the entropy profile. A fourth-moment inequality
absorbs the mean correction in the two-coordinate region, so no mean
partition remains.
At high correlation, convexity reduces the compact
certificate to four endpoint inequalities on each of 24 noise
intervals. For $\atanh\rho\ge8$, an explicit spectral choice and
the local entropy comparison complete the proof analytically.
The stronger marked-coordinate estimates needed for stability are
developed in the companion paper \cite{VuTranStability}.
The table summarizes the ranges covered by each argument. The remaining scalar checks are documented in Appendix~\ref{app:certificate}.

\begin{table}[htbp]\centering\small
\begin{tabular}{@{}lll@{}}
\toprule
Correlation range & Argument & Section\\
\midrule
$[0,3/5]$ & Entropy contraction and a local bound & \ref{ub:lower}\\
$[3/5,457/500]$ & Quadratic Fourier-energy bound & \ref{ub:lower}\\
$[457/500,39/40]$ & Cubic energy or integrated entropy & \ref{sec:cover}\\
$[39/40,49/50]$ & Two-coordinate or integrated entropy bound & \ref{sec:cover}\\
$[49/50,\tanh8]$ & Spectral bound and local entropy & \ref{ub:high-cover}\\
$[\tanh8,1]$ & Analytic spectral and local comparison & \ref{sec:qualitative-completion}\\
\bottomrule
\end{tabular}
\vspace*{3mm}
\caption{Each closed interval is covered for every mean and every dimension.}
\end{table}

\section{Preliminaries and monotone reduction}
\label{sec:preliminaries}

We begin with the representations of a Boolean function that will be used
throughout the proof. Fourier coefficients describe its dependence on sets
of coordinates, while the noise operator describes what an observer can
infer about its value. The connection between these two descriptions lets
us express conditional entropy and its rate of change in terms of the
original function. We follow the uniform-cube conventions in
\cite[Chapters~1--2]{ODonnell2014}; see also
\cite[Chapter~IV]{GarbanSteif2011} for the interpretation of the Fourier
weights as an energy spectrum.

Unless another measure is specified, expectations use the uniform
probability measure on $\{-1,1\}^n$, and $[n]=\{1,\ldots,n\}$. A Boolean function takes values
in $\{-1,1\}$; when it is the input to the noise operator, we also call it
the source. We allow real-valued functions in the definitions below,
since conditional expectations of a Boolean function need not be Boolean.

\subsection{Fourier--Walsh expansion and Fourier levels}

For $S\subseteq[n]$, define the parity function
$\chi_S(x)=\prod_{i\in S}x_i$, with $\chi_\varnothing=1$.
If $S\ne T$, the product $\chi_S\chi_T$ contains a coordinate to an
odd power, so independence and symmetry give $\E[\chi_S\chi_T]=0$.
Each parity has squared norm one. There are $2^n$ parities and the
space of real functions on the cube has dimension $2^n$, so they form
an orthonormal basis for the inner product $\langle h,k\rangle=\E[hk]$.
Consequently every such function has a unique Fourier--Walsh expansion
\[
 h(x)=\sum_{S\subseteq[n]}\widehat h(S)\chi_S(x),
 \qquad \widehat h(S)=\E[h(X)\chi_S(X)].
\]
The coefficient at the empty set is the mean. A singleton coefficient,
written $\widehat h(i)=\widehat h(\{i\})$, measures correlation with
one coordinate, while the coefficients with $|S|=k$ describe the
degree-$k$ part of the multilinear expansion.

For a Boolean source $f$, put $\mu=\E f$. The \emph{Fourier weight
at level $k$} is $W_k(f)=\sum_{|S|=k}\widehat f(S)^2$.
Orthogonality gives Parseval's identity
\[
 \sum_{k=0}^n W_k(f)=\E f^2=1,
 \qquad W_0(f)=\mu^2,\qquad
 \sum_{k\ge1}W_k(f)=\Var(f)=1-\mu^2.
\]
Thus these weights allocate the variance among Fourier degrees.
For example, a dictator has all its weight at level one, whereas
the parity of $k$ coordinates has all its weight at level $k$.
The weight $W_1$ records the combined contribution of all coordinates,
whereas distance to a dictator depends on the largest individual
coefficient. Booleanity links these quantities when the low-degree
weight is close to one: the Friedgut--Kalai--Naor theorem
\cite{FKN2002} says that a Boolean function with Fourier weight at
most $\eta$ above level one is within $C\eta$ in probability of a
constant or a signed coordinate, for an absolute constant $C$.

Write $\alpha(f)=\max_i|\widehat f(i)|$. For each sign $\sigma$,
$\E[f(X)\sigma X_i]=1-2\PP\big(f(X)\ne\sigma X_i\big)$, and therefore
\begin{equation}\label{eq:dictator-distance}
 \delta(f):=\min_{i\in [n],\,\sigma\in\{-1,1\}}
       \PP\big(f(X)\ne\sigma X_i\big)
 =\frac{1-\alpha(f)}2.
\end{equation}
This identity connects Fourier concentration with the distance used
in the stability theorem. We omit the argument $f$ when the source
is fixed and order coordinates by decreasing $|\widehat f(i)|$.
When two coordinates are needed and $n=1$, we append an unused
coordinate; this changes neither entropy nor Fourier weights.

\subsection{Noise, posterior means, and information}

The noise operator averages a function over independent perturbations
of its input. If $Z_i=y_i$ with probability $(1+\rho)/2$ and
$Z_i=-y_i$ otherwise, independently over $i$, then
$T_\rho h(y)=\E[h(Z)]$. The joint distribution of $X$ and $Y_\rho$
is symmetric in its two arguments, so the same formula gives
$T_\rho h(y)=\E[h(X)\mid Y_\rho=y]$.
Independence gives $\E[\chi_S(Z)]=\rho^{|S|}\chi_S(y)$, hence
\[
 T_\rho h=\sum_S\rho^{|S|}\widehat h(S)\chi_S,
 \qquad T_\eta T_\rho=T_{\eta\rho}.
\]
In particular, noise preserves the mean and attenuates each Fourier
degree separately. The second identity expresses the effect of applying
two successive rounds of noise; it follows by multiplying their
factors on every parity. For a Boolean source, Parseval now gives
\[
 \Var(T_\rho f)=\sum_{k\ge1}\rho^{2k}W_k(f).
\]
This formula explains why low-degree information is especially useful
at small correlation, and why discarding higher levels costs more as
$\rho$ approaches one.

Set $g_\rho=T_\rho f$. Since a sign-valued random variable is determined
in distribution by its mean,
$\PP\{f(X)=1\mid Y_\rho=y\}=(1+g_\rho(y))/2$.
Thus $g_\rho$ records the observer's posterior bias: values near zero
mean substantial uncertainty, while values near $\pm1$ mean that the
bit can be predicted with high confidence. By symmetry of the entropy
function in \eqref{eq:entropy-functions}, the posterior entropy at $y$
is $H(g_\rho(y))$. Thus the quantities in \eqref{eq:information-entropy}
can be written as
\begin{equation}\label{eq:posterior-entropy}
 E_\rho(f)=\E H(g_\rho),\qquad
 A_\rho(f)=H(\mu)-E_\rho(f)
          =\E\Phi(g_\rho)-\Phi(\mu).
\end{equation}
The source is \emph{balanced} when $\mu=0$. Its initial entropy is
then $\log2$; an unbalanced source starts with less entropy, which is
why the term $\Phi(\mu)$ must be retained in comparisons with dictators.
Constant sources have zero information and can always be treated directly.

We use the signed gap $G_f$ from \eqref{ub:gap-definitions} for the
extremal argument and the deficit $\Delta_\rho(f)=-G_f(\rho)$ for
stability. Both measure the same comparison, with opposite sign.
For a dictator, $g_\rho(y)=\pm\rho y_i$, so
$E_\rho(f)=H(\rho)$ and $A_\rho(f)=\Phi(\rho)$, confirming the
benchmark. The entropy function is even and concave, decreases on
$[0,1]$, and satisfies $H'(t)=-\atanh t$ on $(-1,1)$.

\subsection{Monotone compression and influences}

An increasing, or monotone, Boolean function satisfies $f(x)\le f(y)$
whenever $x_i\le y_i$ for every coordinate. To reduce to this case,
we sort the two values of $f$ along each coordinate fiber, placing
the smaller value at $-1$ and the larger at $1$. The next lemma
shows that this preserves the mean and can only increase information,
so an upper bound for increasing sources suffices for every source.
This qualitative reduction is due to Courtade and Kumar
\cite[Lemma~2 and Remark~2]{CourtadeKumar2014}. Li and M\'edard
\cite[Theorem~3.2]{LiMedard2021} state the corresponding result for
every convex function of the noisy output. We include the short proof
because the sorting step explains how the arbitrary-mean problem
reduces to increasing sources.

\begin{lemma}[Coordinate compression {\cite{CourtadeKumar2014}}]
\label{lem:compression}
Every $f:\{-1,1\}^n\to\{-1,1\}$ admits an increasing
$g:\{-1,1\}^n\to\{-1,1\}$ with $\E g=\E f$ such that
\[
A_\rho(g)\ge A_\rho(f)
\quad\text{for every }\rho\in[0,1],
\qquad
\widehat g(i)\ge|\widehat f(i)|
\quad\text{for every }i\in[n].
\]
In particular, $\delta(g)\le\delta(f)$.
\end{lemma}

\begin{proof}
Write $f(x_i,z)=u(z)+x_i v(z)$. Sorting the two values on each
$i$-fiber replaces $v$ by $|v|$ and preserves the mean.
For a fixed outside output $y$, put
$b=T_\rho u(y)$, $d=T_\rho v(y)$, and $e=T_\rho|v|(y)$,
where the noise operator acts on the remaining coordinates.
Positivity gives $|d|\le e$. By convexity of $\Phi$, the function
$\frac{\Phi(b+t)+\Phi(b-t)}2$
is even and nondecreasing for $t\ge0$ on its domain.
Its value at $\rho d$ is therefore at most its value at $\rho e$.
Averaging shows that sorting does not decrease $A_\rho$.

The sorted sections are the pointwise minimum and maximum of
the original sections, so sorting preserves monotonicity in
previously sorted coordinates. Sorting each coordinate once
therefore produces an increasing $g$, and the same construction
works for every $\rho\in[0,1]$.

Sorting any coordinate other than $i$ preserves the $i$th
singleton coefficient, while sorting $i$ replaces $\E v$
by $\E|v|\ge|\E v|$. This proves the coefficient comparison
and hence the assertion about $\delta$.
\end{proof}

We keep this compressed source fixed in each differential argument,
so its Fourier coefficients are constants as $\rho$ varies.
Monotonicity also makes the singleton coefficients especially useful,
because they equal coordinate influences, as we now explain.

Let $x^{\oplus i}$ denote $x$ with coordinate $i$ flipped. The
influence of $i$ on a Boolean function is
\[
\operatorname{Inf}_i(f)=\PP\big(f(X)\ne f(X^{\oplus i})\big).
\]
For a real-valued function $h$, its discrete derivative is the half
difference between the two values on an $i$-fiber:
\[
 \partial_i h(x_{-i})=
 \frac{h(x_{-i},1)-h(x_{-i},-1)}2.
\]
For Boolean $f$, this derivative lies in $\{-1,0,1\}$, and its
absolute value indicates whether flipping $i$ changes the output.
Therefore $\operatorname{Inf}_i(f)=\E|\partial_i f|
=\E(\partial_i f)^2$. For increasing $f$ the derivative is
nonnegative, giving $\operatorname{Inf}_i(f)=\E\partial_i f
=\widehat f(i)$. In particular some singleton coefficient is positive
for every nonconstant increasing source. This identification allows
the edge estimates below to retain information about individual
coordinates of the original function.

Two distinct coordinates also satisfy
\begin{equation}\label{eq:pair-bound}
 \widehat f(i)+\widehat f(j)=\E[f(X)(X_i+X_j)]
 \le\E|X_i+X_j|=1\qquad(i\ne j).
\end{equation}
For the ordered nonnegative coefficients of an increasing source,
this implies $\widehat f(2)\le1/2$. We will use such constraints
when deciding which coordinates to retain in an entropy estimate.

\subsection{The Laplacian and entropy production}

To study how entropy changes with the amount of noise, we use the
cube Laplacian. Its coordinate part is
$\mathcal L_i h(x)=[h(x)-h(x^{\oplus i})]/2
=x_i\partial_i h(x_{-i})$, and $\mathcal L=\sum_i\mathcal L_i$.
Thus $\mathcal L_i$ retains precisely the Fourier terms containing
$i$, while $\mathcal L\chi_S=|S|\chi_S$.
The associated Dirichlet energy has both an edge and a Fourier form:
\begin{equation}\label{eq:dirichlet-energy}
 \operatorname{Dir}(h)=\E[h\mathcal Lh]
 =\sum_i\E(\partial_i h)^2
 =\sum_S|S|\widehat h(S)^2.
\end{equation}
It measures variation across cube edges, assigning a Fourier mode
one unit of energy for each coordinate it uses. For Boolean $f$
it is the total influence $\sum_i\operatorname{Inf}_i(f)$.
We write $\operatorname{Dir}_i(h)=\E[h\mathcal L_i h]$ and
$\operatorname{Dir}_K(h)=\sum_{i\in K}\operatorname{Dir}_i(h)$.

Entropy has a corresponding nonlinear energy. For $|g|<1$, define
\begin{equation}\label{eq:entropy-production}
 \mathcal D_i(g)=\E[(\mathcal L_i g)\atanh g],\qquad
 D(g)=\sum_i\mathcal D_i(g).
\end{equation}
On a fiber with values $g_+,g_-$, its contribution is
$(g_+-g_-)(\atanh g_+-\atanh g_-)/4$, which is nonnegative
because $\atanh$ is increasing. These quantities are therefore
suited to measuring entropy production. For a nonconstant Boolean
source and $0<\rho<1$, the strictly positive noise kernel gives
$|T_\rho f|<1$, so all these expressions and derivatives are finite.

Fourier differentiation gives
$\rho\partial_\rho T_\rho f=\mathcal LT_\rho f$. Combining it
with $H'=-\atanh$ yields
\begin{equation}\label{eq:gap-derivative}
 -\rho E'_\rho(f)=D(T_\rho f),\qquad
 \rho G_f'(\rho)=D(T_\rho f)-\rho\atanh\rho.
\end{equation}
Equivalently, in noise time $\tau=-\log\rho$, the operator is
$T_{e^{-\tau}}=e^{-\tau\mathcal L}$ and
$dE_{e^{-\tau}}(f)/d\tau=D(T_{e^{-\tau}}f)$.
Increasing $\tau$ adds noise and creates conditional entropy.
The dictator has production $\rho\atanh\rho$, so the second
identity turns an entropy-production comparison into control of the
slope of the information gap. We often write $u=\atanh\rho$;
this parameter tends to infinity at the noiseless endpoint.

The differential argument uses the following elementary principle.
A positive gap cannot return to a nonpositive endpoint value if its
slope is nonnegative whenever the gap is positive.
\begin{lemma}\label{ub:barrier}
Let $G:[a,b]\to\R$ be continuous, differentiable on $(a,b)$, and
$G(b)\le0$. If $G'(x)\ge0$ whenever $G(x)>0$, then $G\le0$ on $[a,b]$.
\end{lemma}
\begin{proof}
If a positive component exists, $G$ is nondecreasing on that
component, whereas its right endpoint has value zero (or at most
zero if the endpoint is $b$). This is impossible. Continuity covers
the left endpoint.
\end{proof}

At the endpoints, $G_f(0)=0$ and $G_f(1)=-\Phi(\mu)\le0$.
For interior correlations, a hypothetical positive gap gives
$E_\rho(f)<H(\rho)-\Phi(\mu)\le H(\rho)$.
We also use the scalar expansion
\begin{equation}\label{ub:entropy-series}
 \Phi(t)=\sum_{k\ge1}\frac{t^{2k}}{2k(2k-1)},\qquad
 \frac{t^2}{2}\le\Phi(t)\le(\log2)t^2.
\end{equation}
Indeed, differentiating the series gives $\atanh t$, and the value
at zero fixes the integration constant; continuity covers $|t|=1$.
It follows in particular that a positive gap requires
$|\mu|<\sqrt{2H(\rho)}$. Thus sources that could violate the bound
at very large correlation must have small mean. When applying
Lemma~\ref{ub:barrier}, we hold this source fixed on each positive
component. Auxiliary projections and scalar parameters may depend
on $\rho$, but their derivatives do not enter the argument.

\subsection{Finite reversible Markov semigroups}\label{sec:markov-semigroups}

The continuous noise parameter has a similar meaning on a general
finite state space. We explain this setting because the entropy
inequalities in Section~\ref{sec:spectral-completion} use only
reversibility and spectral decomposition. For background on transition
matrices, reversibility, and continuous-time chains, see
\cite[Chapters~1, 12, and 20]{LevinPeres2017}.

Let $\Omega$ be a finite set and let $\nu(x)>0$ for $x\in\Omega$,
with $\sum_x\nu(x)=1$. A transition matrix $p_s(x,y)$ records the
probability of moving from $x$ to $y$ in time $s\ge0$; its entries
are nonnegative and each row sums to one. It acts on a function by
averaging over the possible destinations:
\[
 P_s h(x)=\sum_{y\in\Omega}p_s(x,y)h(y)
         =\E[h(Z_s)\mid Z_0=x].
\]
Here $(Z_s)_{s\ge0}$ is the corresponding continuous-time Markov
chain. The family $(P_s)_{s\ge0}$ is a \emph{Markov semigroup}
when it is continuous in $s$, $P_0=I$, and $P_{s+t}=P_sP_t$.
The last identity says that running the chain for time $s$ and then
for time $t$ has the same effect as running it for time $s+t$.
In particular, these operators preserve constants and positivity.

The measure $\nu$ is \emph{stationary} if
$\sum_x\nu(x)p_s(x,y)=\nu(y)$: starting the chain with distribution
$\nu$ leaves its distribution unchanged. It is \emph{reversible}
with respect to $\nu$ if
$\nu(x)p_s(x,y)=\nu(y)p_s(y,x)$ for every $s,x,y$.
This detailed-balance identity implies stationarity and says that
$P_s$ is self-adjoint for $\langle h,k\rangle_\nu=\E_\nu[hk]$,
where $\E_\nu h=\sum_x\nu(x)h(x)$.

In finite dimension we can write $P_s=e^{-s\mathcal L}$.
Our sign convention makes $\mathcal L$ nonnegative; the usual Markov
generator is $-\mathcal L$. If $q(x,y)\ge0$ is the rate of jumping
from $x$ to $y\ne x$, then
\[
 \mathcal Lh(x)=\sum_{y\ne x}q(x,y)[h(x)-h(y)],
 \qquad \nu(x)q(x,y)=\nu(y)q(y,x).
\]
Pairing the terms for $(x,y)$ and $(y,x)$ gives
\[
 \langle h,\mathcal L k\rangle_\nu
 =\frac12\sum_{x\ne y}\nu(x)q(x,y)
       [h(x)-h(y)][k(x)-k(y)].
\]
Taking $k=h$ shows that
$\operatorname{Dir}(h)=\langle h,\mathcal Lh\rangle_\nu\ge0$.
Reversibility therefore gives an orthonormal eigenbasis with
eigenvalues $\lambda\ge0$. If $h_\lambda$ denotes the projection
of $h$ onto an eigenspace, then
$P_sh=\sum_\lambda e^{-s\lambda}h_\lambda$ and
$\operatorname{Dir}(h)=\sum_\lambda\lambda\|h_\lambda\|_{2,\nu}^2$.
These are the general counterparts of Fourier attenuation and
Fourier energy on the cube.

For the cube, take uniform $\nu$ and let each coordinate flip
independently at rate $1/2$. This gives the Laplacian defined above
and $P_s=T_{e^{-s}}$, so $s=\log(1/\rho)$.
The Walsh characters are its eigenfunctions, with eigenvalues $|S|$.
In the general statements, $\operatorname{Dir}$ and
$D(h)=\E_\nu[(\mathcal Lh)\atanh h]$ use the same definitions
with the specified stationary measure. No irreducibility assumption
is needed for those statements.

\section{Local entropy and Fourier bounds at low correlation}\label{ub:lower}

We prove CK through correlation $0.914$ while developing the local
estimates used in the rest of the paper. The balanced conclusion on
this range is due to Yu \cite{Yu2026Local}; the purpose here is to
keep the mean and the quantitative slack visible in a form compatible
with the later argument. We first retain one or two coordinates,
apply noise on them exactly, and bound the remaining noise by entropy
contraction. For sources without a dominant coordinate, an odd lift
allows Yu's first-level bound and a Khintchine-type estimate to control
the mean together with the level-one weight. A quadratic entropy bound
then completes the interval, with margins that also imply stability.

\subsection{Entropy contraction}

Entropy contraction quantifies how averaging reduces a function's
departure from its mean. For a nonnegative function $v$, write
$\Ent(v)=\E[v\log v]-\E v\log\E v$; if $\E v=1$, this is the
relative entropy of the probability density $v$ with respect to the
uniform measure. Applying it to $1+g$ and $1-g$ converts this functional
inequality into a statement about the conditional entropy of a bit.
We first collect the cube inequalities (see \cite[Sections~2.3, 9.3, and Exercises~10.23--10.24]{ODonnell2014}) needed for that conversion and
for the later influence estimate, using $\|h\|_p=(\E|h|^p)^{1/p}$.

\begin{theorem}[Cube functional inequalities
{\cite{Gross1975,ODonnell2014}}]\label{ext:cube-functional}
For every $h:\{-1,1\}^n\to\R$, with the uniform measure,
$\Var(h)\le\operatorname{Dir}(h)$ and
$\Ent(h^2)\le2\operatorname{Dir}(h)$.
For $0\le\eta\le1$, hypercontractivity gives
$\|T_\eta h\|_2\le\|h\|_{1+\eta^2}$.
\end{theorem}
The following standard consequence
of the cube log-Sobolev inequality is stated in our normalization; see
\cite{Gross1975} and \cite[Exercise~10.23]{ODonnell2014}.

\begin{lemma}[Entropy contraction]\label{jp:contraction}
For $f:\{-1,1\}^n\to[0,\infty)$, write
$\Ent(f)=\E[f\log f]-\E f\log\E f$.
For $0\le\eta\le1$, $\Ent(T_\eta f)\le\eta^2\Ent(f)$.
Consequently, for every $g:\{-1,1\}^n\to[-1,1]$,
\[
 \E H(T_\eta g)\ge \eta^2\E H(g)+(1-\eta^2)H(\E g).
\]
\end{lemma}
\begin{proof}
Entropy tensorizes: $\Ent(f)\le\sum_i\E\Ent_i(f)$,
where $\Ent_i$ denotes entropy on the $i$th coordinate with the other
coordinates fixed. This follows inductively from the entropy chain
rule and convexity of $f\mapsto\Ent(f)$. For $f>0$, convexity
follows from the second variation
$\E[z^2/f]-(\E z)^2/\E f\ge0$, by Cauchy--Schwarz.

On a two-point fiber write $f_\pm=m(1\pm t)$. Its entropy is
$m\Phi(t)$ and its entropy Dirichlet form is $mt\atanh t$.
The power series of $\Phi$ and $\atanh$ give
$t\atanh t\ge2\Phi(t)$ for $|t|<1$. Now set
$f_\tau=T_{e^{-\tau}}f$. Integration by parts on the cube yields
\[
 -\frac{d}{d\tau}\Ent(f_\tau)
 =\sum_i\E[\mathcal L_i f_\tau\,\mathcal L_i\log f_\tau]
 \ge2\sum_i\E\Ent_i(f_\tau)
 \ge2\Ent(f_\tau),
\]
Integration proves the contraction inequality. Nonnegative $f$ is treated by applying
the result to $f+\epsilon$ and taking $\epsilon\downarrow0$;
the cases $\eta=0,1$ follow by continuity.

Finally, apply the contraction inequality to $1+g$ and $1-g$ and use
$\Ent(1+g)+\Ent(1-g)
 =2\{H(\E g)-\E H(g)\}$.
This gives the posterior-entropy bound.
\end{proof}

Let $T_\eta^K$ denote noise applied only to coordinates in a fixed core
$K$, and define $f_K:\{-1,1\}^K\to[-1,1]$ by
$f_K=\E[f\mid X_K]$. First apply noise on $K$; for
each core output, apply Lemma~\ref{jp:contraction} to the remaining
coordinates. Averaging over core outputs gives
\begin{equation}\label{jp:core-contraction}
 E_\eta(f)\ge \eta^2\E H(T_\eta^K f)
 +(1-\eta^2)\E_{X_K}H(T_\eta^K f_K).
\end{equation}
This also holds when the outside cube is a singleton.

\begin{corollary}
\label{ub:information-contraction}
For every $f:\{-1,1\}^n\to\{-1,1\}$ and
$0\le\eta,\rho_0\le1$,
$A_{\eta\rho_0}(f)\le\eta^2 A_{\rho_0}(f)$. In particular,
$A_\rho(f)\le\rho^2 H(\mu)$.
If $|\mu|\ge61/100$, the CK inequality holds at every correlation.
\end{corollary}
\begin{proof}
For $g:\{-1,1\}^n\to[-1,1]$ with $\E g=\mu$,
$\tfrac12[\Ent(1+g)+\Ent(1-g)]
 =\E\Phi(g)-\Phi(\mu)$.
Apply Lemma~\ref{jp:contraction} to $1\pm T_{\rho_0}f$ and
use the semigroup property. Taking $\rho_0=1$ gives the second
inequality. The entropy series and $\log2<347/500$ give
$H(61/100)<347/500-(61/100)^2/2-(61/100)^4/12<1/2$;
now use $\Phi(\rho)\ge\rho^2/2$.
\end{proof}

Consequently, a bound $A_{\rho_0}(f)\le\rho_0^2/2$ at one positive
correlation proves CK for the same source at every smaller correlation:
contraction gives $A_\rho(f)\le\rho^2/2\le\Phi(\rho)$. This is how
the estimate at $\rho=3/5$ will cover the interval down to zero.

\subsection{An exact one-coordinate inequality at arbitrary mean}

The balanced specialization of the next estimate is Yu's
one-coordinate bound \cite[Theorem~4.6 and Remark~4.7]{Yu2026Local}.
His Section~6 gives a broader framework for arbitrary means. We give
the direct entropy-contraction proof of the form needed here, including
the convex mean correction that permits comparison with a dictator.
For $0\le a\le1$, define
\begin{equation}\label{eq:local-deficit}
 M_\rho(a)=(1-\rho^2)H(\rho a)-(1-\rho^2a)H(\rho).
\end{equation}
This is the lower bound for the information deficit
obtained by retaining a coefficient of magnitude $a$. Its value is
zero at $a=1$, and concavity will allow one verified value to control
an entire neighborhood of the dictators.
\begin{theorem}[Local entropy comparison]\label{ub:local}
Let $f:\{-1,1\}^n\to\{-1,1\}$, $i\in[n]$, and $0\le\rho\le1$.
Then
\[
 E_\rho(f)\ge\frac{1-\rho^2}{2}
 \bigl[H(\mu+\rho\widehat f(i))+H(\mu-\rho\widehat f(i))\bigr]
 +\rho^2|\widehat f(i)|H(\rho).
\]
Moreover, $\Phi(\rho)-A_\rho(f)\ge M_\rho(|\widehat f(i)|)$.
For fixed $0<\rho<1$, $M_\rho$ is strictly concave on $[0,1]$
and $M_\rho(1)=0$. Consequently, if $M_\rho(A)\ge0$ for some $A\in[0,1]$, then
$A_\rho(f)\le\Phi(\rho)$ for every Boolean function $f$ with
$\max_{i\in[n]}|\widehat f(i)|\ge A$.
\end{theorem}
\begin{proof}
Orient coordinate $i$ so that $a=\widehat f(i)\ge0$.
The restriction to coordinate $i$ is constant or a signed dictator.
Its averaged entropy after noise is $\operatorname{Inf}_i(f)H(\rho)$,
and $\operatorname{Inf}_i(f)\ge a$. The conditional mean is
$\mu+aX_i$. Formula~\eqref{jp:core-contraction} therefore gives the first inequality. Feasibility of the two conditional means
implies $|\mu|+a\le1$.

Define on this feasible interval
\[
 B_{\rho,a}(\mu)=\Phi(\mu)+(1-{\rho^2})
 \left\{\frac{H(\mu+\rho a)+H(\mu-\rho a)}2-H(\rho a)\right\}.
\]
This is even and $B_{\rho,a}(0)=B'_{\rho,a}(0)=0$.
The feasibility condition gives $a\le1-\mu$ and $a\le1+\mu$.
These two inequalities control the two reciprocal terms in the entropy
curvature. Indeed, since
$-H''(t)=\tfrac12[(1-t)^{-1}+(1+t)^{-1}]$, we have
\[
 -\frac{H''(\mu+\rho a)+H''(\mu-\rho a)}2
 =\frac12\sum_{\sigma\in\{-1,1\}}
   \frac{1-\sigma\mu}{(1-\sigma\mu)^2-\rho^2a^2}
 \le\frac{1}{(1-\rho^2)(1-\mu^2)}.
\]
Each denominator is bounded below by
$(1-\rho^2)(1-\sigma\mu)^2$. It follows that $B''(\mu)\ge0$,
so evenness and $B(0)=0$ give $B\ge0$. Adding
$\Phi(\mu)-H(\rho)$ to the first inequality now proves the second.
The calculation is for interior parameters; continuity covers the boundary.
Finally $M''_\rho(a)=-(1-{\rho^2}){\rho^2}/(1-{\rho^2}a^2)<0$, proving the last claim.
\end{proof}

\begin{lemma}\label{ub:local-propagation}
For fixed $a\in[0,1]$, if $M_{\rho_0}(a)\ge0$ at some
$0<\rho_0<1$, then
\[
 \frac{M_\rho(a)}{\rho^2}\ge
 \frac{M_{\rho_0}(a)}{\rho_0^2}\ge0
 \qquad(0<\rho\le\rho_0).
\]
In particular, $M_\rho(a)\ge0$ for $0\le\rho\le\rho_0$.
\end{lemma}
\begin{proof}
Suppose first that $0<a<1$, and set $F(\rho)=M_\rho(a)/\rho^2$. With $c_n=1/[2n(2n-1)]$, the entropy
series gives
\[
 F(\rho)=(1-a)\left(\frac{1+a}{2}-{\log2}\right)+\sum_{n\ge1}d_n\rho^{2n},
 \qquad
 d_n=c_{n+1}(1-a^{2n+2})-c_n(a-a^{2n}).
\]
The coefficient sign is determined by the ratio of its two positive
terms. With $b(x)=x/(1-a^x)$, that ratio is
\[
 R_n=\frac{2n}{a(2n+1)}\frac{b(2n-1)}{b(2n+2)},
 \qquad \operatorname{sign}(d_n)=\operatorname{sign}(R_n-1).
\]
Writing $t=-\log a>0$, we have
$(\log b)''(x)=-x^{-2}+t^2/[4\sinh^2(tx/2)]<0$.
Thus $b(x)/b(x+3)$ increases strictly, and so does $R_n$.
Since $R_n\to1/a>1$, the coefficients $d_n$ change sign at
most once, from negative to positive.

The derivative series has the same sign pattern. If every coefficient
is nonnegative, then $F$ increases strictly. Otherwise let $m$ be
the last index with $d_m<0$. The series $F'(\rho)/\rho^{2m-1}$
increases strictly: its negative terms have nonpositive powers and
its positive terms have positive powers. Termwise differentiation
is valid on compact subintervals of $(0,1)$. Hence $F$ first decreases
and then increases, with either part possibly empty.
Since $F(1)=0$ by continuity, a nonnegative interior value lies on
the decreasing part. Thus $F(\rho)\ge F(\rho_0)$ for
$0<\rho\le\rho_0$, which proves the quantitative assertion. At $a=1$, $M_\rho(a)=0$ identically;
at $a=0$, $M_\rho(a)=\Phi(\rho)-\rho^2{\log2}<0$ for $0<\rho<1$,
so the hypothesis is vacuous. Finally, $M_0(a)=0$.
\end{proof}

Together with concavity in $a$, this means that one check of
$M_{\rho_0}(a_0)\ge0$ covers every smaller correlation and every
larger singleton coefficient. The endpoint $\rho_0=1$ is excluded:
$M_1(a)=0$ for every $a$, so that endpoint alone gives no information.

The following explicit region will remove the local computations from
the compact high-correlation interval.

\begin{lemma}\label{ub:exponential-local}
Let $f:\{-1,1\}^n\to\{-1,1\}$. If $u\ge\log9$ and
$\delta(f)\le e^{-u}$, then CK holds at
$\rho=\tanh u$.
\end{lemma}
\begin{proof}
Set $\delta=e^{-u}\le1/9$, $p=\delta^2/(1+\delta^2)$, and
$a=1-2\delta$. Then $\rho=1-2p$ and
$H(\rho a)=H_{\rm b}(\delta+p(1-2\delta))$.
By Theorem~\ref{ub:local} and concavity in the singleton coefficient,
it suffices to show $M_\rho(a)>0$.

The elementary bounds $H_{\rm b}(\delta)/\delta\ge u+1-\delta$ and
$H_{\rm b}(p)/p\le2u+\delta^2+1$ follow from
$x\le-\log(1-x)\le x/(1-x)$. Throughout
$\delta\le x\le\delta+\delta^2$, the same logarithm bounds give
$H_{\rm b}'(x)\ge u-7\delta/3>0$. Integrating from $\delta$ and using
$p\ge\delta^2(1-\delta^2)$ therefore gives
\[
 \frac{M_\rho(a)}{2\delta p}
 \ge 2(1-\delta^2)
 \left[u+1-\delta+\delta(1-\delta^2)(1-2\delta)
                      (u-7\delta/3)\right]
 -(1+2\delta-4\delta^2)(2u+\delta^2+1).
\]
After expansion, dropping positive terms leaves
$1-2u\delta-4u\delta^3-4u\delta^6-4\delta
 -11\delta^2/3-56\delta^5/3-14\delta^6/3$.
Each $\delta^k\log(1/\delta)$ increases on $(0,1/9]$.
Using $\log9<11/5$, this lower bound is at least
$71636/7971615>0$. Finally, $\log9<11/5$ follows from
$e>163/60$ and $e^{1/5}>61/50$, by the exponential series.
\end{proof}

\subsection{Two-coordinate sections with the actual mean}

The last intermediate correlation band needs two coordinates.  After
compression, every two-dimensional section is one of the six functions
in Table~\ref{tab:two-coordinate-sections}.  This finite list gives a
sharper core entropy bound, while the actual mean and interaction remain
inside a simple feasible polygon.

For an increasing source, the two largest singleton coefficients
satisfy $\widehat f(1)\ge\widehat f(2)\ge0$. The conditional mean
on these coordinates is
\begin{equation}\label{ub:core-mean}
 f_{\{1,2\}}(x,y)=\mu+\widehat f(1)x+\widehat f(2)y
                         +\widehat f(\{1,2\})xy.
\end{equation}
Let $J_{\rm AND}(\rho)$ be the average noisy entropy of the
two-coordinate AND function. With $p_0=(1-\rho)/2$, it equals
\begin{equation}\label{eq:and-entropy}
 J_{\rm AND}(\rho)=
 \frac{H_{\rm b}(p_0^2)+2H_{\rm b}(p_0(1-p_0))+H_{\rm b}(2p_0-p_0^2)}4.
\end{equation}

\begin{proposition}\label{jp:local2-exact}
Let $f:\{-1,1\}^n\to\{-1,1\}$ be increasing. For $0\le\rho\le1$,
\[
 E_\rho(f)\ge\rho^2\bigl[\widehat f(1)H(\rho)
            +\widehat f(2)[2J_{\rm AND}(\rho)-H(\rho)]\bigr]
 +(1-\rho^2)\E H(T_\rho f_{\{1,2\}}).
\]
\end{proposition}
\begin{proof}
Write $a=\widehat f(1)$ and $b=\widehat f(2)$ in this proof.
First, $J_{\rm AND}(\rho)\le H(\rho)$ at every correlation.
For independent Bernoulli variables $U,V$ with means $s,t$, put
$Z=UV$. The Shannon entropy identity
$H(Z\mid U)+H(Z\mid V)-H(Z)=I(U;V\mid Z)\ge0$
gives $H_{\rm b}(st)\le tH_{\rm b}(s)+sH_{\rm b}(t)$.
Apply this conditionally on the two noisy observations and average.
The posterior means are independent, each averages to $1/2$,
and both posterior entropies equal $H(\rho)$, proving the claim.

Every monotone two-coordinate section is listed in
Table~\ref{tab:two-coordinate-sections}.
\begin{table}[htbp]\centering\small
\begin{tabular}{@{}lccccc@{}}
\toprule
Section & Mean & Singleton 1 & Singleton 2 & Interaction & Entropy\\
\midrule
$+1$ & $1$ & $0$ & $0$ & $0$ & $0$\\
$-1$ & $-1$ & $0$ & $0$ & $0$ & $0$\\
$x$ & $0$ & $1$ & $0$ & $0$ & $H(\rho)$\\
$y$ & $0$ & $0$ & $1$ & $0$ & $H(\rho)$\\
AND & $-1/2$ & $1/2$ & $1/2$ & $1/2$ & $J_{\rm AND}(\rho)$\\
OR & $1/2$ & $1/2$ & $1/2$ & $-1/2$ & $J_{\rm AND}(\rho)$\\
\bottomrule
\end{tabular}
\vspace*{3mm}
\caption{Fourier data and averaged noisy entropy of each monotone section.}
\label{tab:two-coordinate-sections}
\end{table}

Let $p_x$ and $p_y$ be the probabilities, over the outside
coordinates, of the sections $x$ and $y$, respectively, and let
$q$ be the total probability of an AND or OR section. The table
gives $a=p_x+q/2$ and $b=p_y+q/2$, so $q\le2b$. The averaged
core entropy is $(a+b)H(\rho)-q[H(\rho)-J_{\rm AND}(\rho)]$,
which is at least $aH(\rho)+b[2J_{\rm AND}(\rho)-H(\rho)]$.
Now apply \eqref{jp:core-contraction} and \eqref{ub:core-mean}.
\end{proof}

The conditional mean is increasing and takes values in $[-1,1]$.
Its edge differences and extreme values imply
$\widehat f(1)+\widehat f(2)\le1$,
$|\widehat f(\{1,2\})|\le\widehat f(2)$, and
$|\mu|\le1-\widehat f(1)$.
We will control its mean through moments, without partitioning the
feasible mean--interaction domain.

\begin{lemma}[Fourth-moment mean compensation]\label{lem:mean-moments}
Let $(\Omega,\nu)$ be a probability space,
$g:\Omega\to[-1,1]$, and $h=g-\E_\nu g$. For $0\le\lambda\le1/20$,
\[
 \lambda\E_\nu H(g)+\Phi(\E_\nu g)
 \ge\lambda\left[\log2-\frac{\E_\nu h^2}{2}
          -(\log2-\tfrac12)\E_\nu h^4\right]
       -\frac25\lambda^2(\E_\nu h^3)^2.
\]
\end{lemma}
\begin{proof}
The entropy series gives
$H(t)\ge\log2-t^2/2-(\log2-1/2)t^4$ and
$\Phi(t)\ge t^2/2+t^4/12$ on $[-1,1]$.
Put $\mu=\E_\nu g$ and expand the second and fourth moments of
$\mu+h$. After subtracting the first term on the right, the
remaining lower bound is
\[
 \left[\frac{1-\lambda}{2}
       -6\lambda(\log2-\tfrac12)\E_\nu h^2\right]\mu^2
 -4\lambda(\log2-\tfrac12)\mu\E_\nu h^3
 +[\tfrac1{12}-\lambda(\log2-\tfrac12)]\mu^4.
\]
Since $\E_\nu h^2\le1$ and $\log2-1/2<1/5$, the quadratic
coefficient exceeds $2/5$ and the quartic coefficient is nonnegative.
Completing a square in $|\mu|$ gives the error
$-(2/5)\lambda^2(\E_\nu h^3)^2$.
\end{proof}

The following elementary entropy bounds will make the resulting
moment criterion polynomial.

\begin{lemma}\label{lem:and-rational}
For $39/40\le\rho\le49/50$,
$H(\rho)/(1-\rho^2)<17/12$ and
$J_{\rm AND}(\rho)\ge H(\rho)-(1-\rho^2)/10$.
\end{lemma}
\begin{proof}
The entropy series makes $H(\rho)/(1-\rho^2)$ increasing.
At $\rho=49/50$, use
$H_{\rm b}(p)\le p[1+\log(1/p)]$ with $p=1/100$ and
$\log100<461/100$ to get $17/12$.

For the second bound put $p=(1-\rho)/2\le1/80$.
Concavity of $H_{\rm b}$ gives
$H_{\rm b}(p^2)\ge pH_{\rm b}(p)$,
$H_{\rm b}(p-p^2)\ge(1-p)H_{\rm b}(p)$, and
$H_{\rm b}(2p-p^2)\ge(1-p/2)H_{\rm b}(2p)$.
The second derivative of $-(1-p)\log(1-p)$ gives
$H_{\rm b}(2p)\ge2H_{\rm b}(p)-2p\log2-p^2/(1-2p)$.
Substituting these bounds in \eqref{eq:and-entropy} yields
\[
 H(\rho)-J_{\rm AND}(\rho)
 \le p\left[\frac{H_{\rm b}(p)+\log2}{2}
                  +\frac{p}{4(1-2p)}\right].
\]
Now $H_{\rm b}(p)<7/100$, using $\log80<9/2$, and $\log2<7/10$.
The bracket is at most $7/200+7/20+1/312<79/200$
and hence is less than $(2/5)(1-p)$, as required.
The exponential series through degree 11 proves both logarithm
bounds used here; the logarithm series gives $9/13<\log2<7/10$.
\end{proof}

Two easy source classes are already controlled: large $|\mu|$ by
information contraction and large $\alpha$ by the local theorem.  For
the remaining class we need a dimension-free upper bound on
$W_1+\mu^2$.  The following lift turns the mean into one more singleton
coefficient, so balanced Fourier estimates apply without discarding it.

\subsection{Retaining the mean as a singleton coefficient}

The entropy estimate below depends on $W_1(f)+\mu^2$. To control
these two quantities together, we place the mean into the first level
of an auxiliary balanced function on one more coordinate. The
resulting lift lets us apply Yu's balanced Fourier estimate directly
to both the original singleton coefficients and the mean. This use
of the lift concerns Fourier weights; no comparison between the
mutual informations of the original and lifted functions is required.

For $0\le a\le1$, define
\begin{equation}\label{eq:yu-fourier-cap}
 Q_{\rm LP}(a)=\begin{cases}
 a^2+1-a,&0\le a\le1/2,\\
 a^2+2\sqrt2(1-a)^{3/2}-2(1-a)^2,&1/2\le a\le1.
 \end{cases}
\end{equation}

\begin{lemma}[Yu's first-level bound
{\cite[Proposition~2 and Remark~2]{Yu2023}}]\label{ext:yu-bound}
For every $g:\{-1,1\}^n\to\{-1,1\}$ with $\E g=0$ and every
$i\in[n]$, $W_1(g)\le Q_{\rm LP}(|\widehat g(i)|)$.
The coordinate $i$ need not have the largest singleton coefficient.
\end{lemma}
Yu states the result for $0/1$-valued functions. Applying it to
$(1+g)/2$ gives the sign-valued normalization above.

\begin{theorem}[K\"onig--Sch\"utt--Tomczak-Jaegermann
{\cite{KonigSchuttTomczak1999}}]
For independent uniform signs $X_1,\ldots,X_n$ and real
$c_1,\ldots,c_n$,
\[
 \E\left|\sum_i c_iX_i\right|
 \le\sqrt{2/\pi}\,\|c\|_2+(1-\sqrt{2/\pi})\|c\|_\infty.
\]
\end{theorem}

Let $g_1:\{-1,1\}^n\to\R$ be the degree-one part of $g$.
Then $W_1(g)=\E[gg_1]\le\E|g_1|$. Solving the resulting quadratic
in $\sqrt{W_1(g)}$ gives
\begin{equation}\label{ub:Khintchine}
 W_1(g)\le K_{\rm Kh}(\alpha(g)),\qquad
 K_{\rm Kh}(a)=
 \frac{(1+\sqrt{1+2(\pi-\sqrt{2\pi})a})^2}{2\pi}.
\end{equation}
This inequality does not require balance.

\begin{lemma}[Odd-lift bound]\label{ub:lift}
For every $f:\{-1,1\}^n\to\{-1,1\}$, write
$\mu=\E f$, $\alpha=\max_i|\widehat f(i)|$, and
$W_1=\sum_i\widehat f(i)^2$. With the bounds defined in
\eqref{eq:yu-fourier-cap} and \eqref{ub:Khintchine},
\[
 W_1+\mu^2\le\min\{K_{\rm Kh}(\max(\alpha,|\mu|)),
                   Q_{\rm LP}(\alpha),Q_{\rm LP}(|\mu|)\}.
\]
In particular, if $\max(\alpha,|\mu|)\le69/100$, then
\[
 W_1+\mu^2<\frac{193027}{250000}.
\]
\end{lemma}
\begin{proof}
Let $x_0$ be an additional uniform sign, and define
$g(x_0,x)=x_0f(x_0x)$, with $x_0$ multiplying every coordinate of $x$.
This function is odd under simultaneous sign reversal, hence
balanced. Changing variables $y=x_0x$ gives
$\widehat g(0)=\mu$ and
$\widehat g(i)=\widehat f(i)$ for $i\ge1$.
Thus $W_1(g)=W_1(f)+\mu^2$ and
$\alpha(g)=\max(\alpha,|\mu|)$.
Apply Lemma~\ref{ext:yu-bound} to both retained coefficients and
\eqref{ub:Khintchine} to the maximum.

For the numerical constant, set $x=\max(\alpha,|\mu|)$ and
$b_*=87831/250000$.  On $0\le x\le b_*$ we use the increasing
bound $K_{\rm Kh}(x)$.  On $b_*\le x\le1/2$ we use the decreasing
bound $Q_{\rm LP}(x)=x^2+1-x$.  Finally, on
$1/2\le x\le69/100$, the function $Q_{\rm LP}$ is convex because
$Q_{\rm LP}''(x)=-2+3/[2\sqrt{(1-x)/2}]>0$, so its maximum occurs
at an endpoint.  It therefore suffices to check
$K_{\rm Kh}(b_*),Q_{\rm LP}(b_*),Q_{\rm LP}(1/2),
Q_{\rm LP}(69/100)<{\frac{193027}{250000}}$.
The rational and Arb checks give respective margins exceeding
$2.5920\cdot10^{-6}$, $3.4470\cdot10^{-6}$, $0.022108$, and
$1.9511\cdot10^{-5}$. The increasing, decreasing and convex
pieces cover $[0,69/100]$, proving the stated Fourier cap.
\end{proof}

\subsection{A scalar entropy inequality and a quadratic energy bound}

We now combine the Fourier cap with a pointwise lower bound for
$H(|g|)$, where $g=T_\rho f$. The chosen quadratic factor has an
expectation that can be controlled by $W_1$ and the mean, leaving a
one-variable tangent inequality for verification. The following
Fourier truncation also explains the cubic estimate used later:
choosing a larger degree keeps more low-level information before
bounding the remaining levels together.

\begin{lemma}\label{ub:degree-energy}
Let $f:\{-1,1\}^n\to\{-1,1\}$ and $g=T_\rho f$,
where $0\le\rho\le1$. For every integer $d\ge1$,
\[
 \E g^2\le \rho^d\E[fg]+(1-\rho^d)\mu^2
   +\sum_{k=1}^{d-1}(\rho^{2k}-\rho^{d+k})W_k,
\]
and consequently
\[
 \E[(1-|g|)(1+|g|-\rho^d)]
 \ge (1-\rho^d)(1-\mu^2)
   -\sum_{k=1}^{d-1}(\rho^{2k}-\rho^{d+k})W_k.
\]
\end{lemma}
\begin{proof}
In the first inequality, the Fourier coefficients agree through
level $d-1$, while on every level $k\ge d$ one has
$\rho^{2k}\le\rho^{d+k}$.  Next expand the left side of
the second inequality as
$1-\rho^d+\rho^d\E|g|-\E g^2$ and use
$\E[fg]\le\E|g|$.
\end{proof}

Put $g=T_\rho f$. Taking $d=2$ in Lemma~\ref{ub:degree-energy} gives
\begin{equation}\label{ub:quadratic-product}
 \E[(1-|g|)(1+|g|-{\rho^2})]
 \ge(1-{\rho^2})(1-\mu^2)-{\rho^2}(1-\rho)W_1.
\end{equation}
Assume now that $\max(\alpha,|\mu|)\le69/100$, so
$W_1\le{\frac{193027}{250000}}-\mu^2$ by Lemma~\ref{ub:lift}.  If a scalar $c>0$ satisfies
\begin{equation}\label{ub:quadratic-gate}
 H(t)\ge c(1-t)(1+t-{\rho^2})\qquad(0\le t\le1),
\end{equation}
then Lemma~\ref{ub:lift} and \eqref{ub:quadratic-product} imply
\begin{equation}\label{ub:quadratic-budget}
 \begin{split}
 \Phi(\rho)-A_\rho(f)\ge{}&
 c[(1-{\rho^2})-{\rho^2}(1-\rho){\frac{193027}{250000}}]-H(\rho)\\
 &+\Phi(\mu)-c(1-2{\rho^2}+\rho^3)\mu^2.
 \end{split}
\end{equation}

\begin{lemma}\label{ub:tangent}
For $c>0$ and $b<1$, the function
$\chi(t)=-2H(t)/(1-t)+2c(1+t-b)$ is concave on $[0,1)$. If its tangent at some
$t_0\in(0,1)$ is nonpositive at both $0$ and $1$, then
$H(t)\ge c(1-t)(1+t-b)$ on $[0,1]$.
\end{lemma}
\begin{proof}
For $p=(1-t)/2$ and $q(p)=H_{\rm b}(p)/p$,
$q''(p)=[-p/(1-p)-2\log(1-p)]/p^3\ge0$: the numerator vanishes at zero and has derivative
$(1-2p)/(1-p)^2\ge0$ for $0<p\le1/2$. Thus $-q((1-t)/2)$
is concave. A concave function lies below each tangent, and a
linear function nonpositive at both endpoints is nonpositive
throughout. Multiplication by $(1-t)/2$ proves the assertion for
$t<1$; continuity gives $t=1$. For evaluation one may use
$\chi'(t_0)=\log(1-p_0)/(2p_0^2)+2c$, where $p_0=(1-t_0)/2$.
\end{proof}

\subsection{Completing the low-correlation bound}

An entropy bound whose reciprocal coefficient is affine in $\rho^2$
can be interpolated from two correlations. This removes the need to
partition the central range. A decreasing energy remainder and a
positive mean correction then give the information bound; entropy
contraction covers all smaller correlations.

\begin{proposition}\label{ub:lower-interval}
Theorem~\ref{jp:main} holds on $0\le\rho\le457/500$.
\end{proposition}
\begin{proof}
For $3/5\le\rho\le457/500$, use
$c=57/(76-49\rho^2)$ in
\eqref{ub:quadratic-gate}. After multiplication by $1/c$, this
inequality is affine in $\rho^2$, so it suffices to check the two
correlation endpoints. Lemma~\ref{ub:tangent} does so with
$t_0=39/100$ at $\rho=3/5$ and $t_0=829/1250$ at
$\rho=457/500$. All four tangent values are less than $-1/2000$,
as checked by the explicit logarithm series in the appendix.

Lemma~\ref{app:central-energy} makes the constant term in
\eqref{ub:quadratic-budget} positive. The coefficient
$c(1-2\rho^2+\rho^3)$ decreases with $\rho$: its derivative
has the sign of
$-206+228\rho-49\rho^3$, which increases to a
negative value at $\rho=1$. At $\rho=3/5$ the coefficient is
$3534/7295<49/100$. Thus the mean correction in
\eqref{ub:quadratic-budget} is at least $\mu^2/100$.
This proves CK throughout the central interval when
$\alpha\le69/100$ and $|\mu|\le61/100$.

At $\rho=3/5$ the same energy bound gives
$A_{3/5}(f)<9/50=(3/5)^2/2$, using $c=1425/1459$ and
$\log2<347/500$.
Corollary~\ref{ub:information-contraction} therefore proves CK
for these sources at every smaller correlation.

For the remaining sources, the same corollary covers
$|\mu|\ge61/100$ at every correlation. The fixed comparison
$M_{457/500}(69/100)>0.0036195$, together with
Lemma~\ref{ub:local-propagation} and Theorem~\ref{ub:local},
covers $\alpha\ge69/100$ on the whole interval.
Thus all sources are covered.
\end{proof}

\section{Scalar entropy inequalities and profile clocks}\label{sec:profile-clocks}

The low-correlation argument retained total Fourier weights. We now
keep the individual singleton coefficients, which record how the
first-level dependence is distributed among coordinates. The bridge
to entropy is an inequality on one edge: its production can be bounded
below using only its average entropy and the half difference of its
endpoint values. We encode that bound in a scalar function $B_e(q)$.
Convexity will allow us to average edges in each coordinate direction,
and then to integrate the resulting production bound.

\subsection{The scalar entropy profile}

Consider first an edge with posterior values $-t,t$, where $0<t<1$.
Its average entropy is $H(t)$, its half difference is $t$, and its
entropy production is $t\atanh t$. Writing $v=\atanh t$, we denote the ratio of entropy to edge size
by $F(v)$ and use its inverse to define the production profile:
for $v,e,q>0$, put
\begin{equation}\label{eq:entropy-profile}
 F(v)=\frac{H(\tanh v)}{\tanh v},\qquad B_e(q)=qF^{-1}(e/q).
\end{equation}
The next lemma shows that $F$ decreases bijectively from
$(0,\infty)$ onto $(0,\infty)$, so the inverse is well defined.
Equivalently, choose the unique $t\in(0,1)$ satisfying $et=qH(t)$
and set $B_e(q)=q\atanh t$, with $B_e(0)=0$ by continuity.
For an actual centered edge this recovers its production exactly;
Lemma~\ref{lem:edge-entropy-profile} proves that it is a lower bound
for every edge with the prescribed entropy and half difference.
The convexity properties established first are what allow that
pointwise inequality to survive averaging.

\begin{lemma}[Profile convexity]\label{as:profile-convex}
The function $F$ is strictly decreasing and strictly convex on
$(0,\infty)$, while $v\mapsto vF(v)$ is strictly decreasing and
strictly log-concave. Consequently $1/[vF(v)^r]$ is convex for
$0\le r\le1$. Also $-(\log F)'(v)>1$ everywhere,
and this derivative exceeds $7/4$ for $v\ge4$.
\end{lemma}
\begin{proof}
Put $t=\tanh v$ and $B(v)=\log(2\cosh v)$. Direct differentiation gives
\[
 F'=-\frac{(1-t^2)B}{t^2}<0,\qquad
 F''=\frac{(1-t^2)(2B-t^2)}{t^3}>0,
\]
since $B\ge\log2>1/2$. The endpoint limits of $F$ are $+\infty$ and $0$.
Now set
$\mathcal H(v)=\cosh v\,H(\tanh v)=B(v)\cosh v-v\sinh v$.
Then $\mathcal H'=B\sinh v-v\cosh v$ and
$\mathcal H''=\mathcal H-\operatorname{sech}v$, so
$\mathcal H\mathcal H''-(\mathcal H')^2=B^2-v^2-H(\tanh v)$.
The derivative of the right side is
$2tH(t)-(1-t^2)\atanh t$ at $t=\tanh v$.
As a function of $t\in[0,1]$, this has zero endpoint values and
second derivative $-2\atanh t<0$ in the interior. It is therefore
positive. Since $B^2-v^2-H(\tanh v)$ tends to zero as $v\to\infty$,
it is negative everywhere, proving strict log-concavity of $\mathcal H$.
Its even extension has derivative zero at zero, so $\mathcal H'<0$.
Since $F=\mathcal H/\sinh v$, this also gives
$-F'/F=\coth v-\mathcal H'/\mathcal H>1$.
Now $vF(v)=(v/\sinh v)\mathcal H(v)$, and
$(\log(v/\sinh v))''=-1/v^2+1/\sinh^2v<0$.
The product is even at zero with derivative zero, hence strictly
decreasing on $(0,\infty)$.
Finally $1/[vF(v)^r]=v^{r-1}[vF(v)]^{-r}$ is log-convex.

For the derivative bound use
$-F'/F=B/[\sinh v\cosh v\,H(\tanh v)]$, $B>v$, and the identity
$H(\tanh v)=\log(1+e^{-2v})+2v/(e^{2v}+1)
\le e^{-2v}(1+2v)$.
This gives $-F'/F>4v/[(1-e^{-4v})(1+2v)]>4v/(1+2v)$.
The bounds $x/(1+x)\le\log(1+x)\le x$ also give
\begin{equation}\label{rt:scaled-profile}
 2v+1\le e^{2v}F(v)
 \le\frac{1+e^{-2v}}{1-e^{-2v}}(2v+1).\qedhere
\end{equation}
\end{proof}

\begin{lemma}\label{u963:scalar}
The function $B_e(q)$ increases with $q$ and decreases with $e$.
For each $e>0$, the map $v\mapsto B_e(\sqrt v)$ is increasing and
strictly concave on $[0,\infty)$. The map
$(e,q)\mapsto B_e(|q|)$ is jointly convex on $(0,\infty)\times\R$.
Moreover, $B_e(cq)=cB_{e/c}(q)$ for $c>0$.
\end{lemma}

\begin{proof}
The defining formula gives monotonicity and scaling.
Put
\begin{equation}\label{eq:profile-elasticity}
 \mathcal K(v)=-\frac{vF'(v)}{F(v)}.
\end{equation}
If $j(v)=\log(vF(v))$, Lemma~\ref{as:profile-convex} gives
$j',j''<0$, hence $\mathcal K=1-vj'>1$ and
$\mathcal K'=-j'-vj''>0$.
For fixed $e$, parametrize $q=e/F(v)$ and put $x=q^2$.
Then
\[
 \frac{d}{dx}B_e(\sqrt x)
 =\frac{vF(v)}{2e}\left(1+\frac1{\mathcal K(v)}\right).
\]
Both factors are positive and strictly decreasing in $v$, while
$x$ increases. This proves strict concavity for $x>0$.
The expansion $B_e(\sqrt x)=(\log2/e)x+O(x^2)$ extends it to zero.

Since $F$ is decreasing and convex, its inverse is convex.
Thus $B_e(q)=qF^{-1}(e/q)$ is jointly convex for $q>0$:
if $q=\theta q_1+(1-\theta)q_2$, write
$[\theta e_1+(1-\theta)e_2]/q$ as the convex combination of
$e_1/q_1,e_2/q_2$ with weights $\theta q_1/q,(1-\theta)q_2/q$,
and apply convexity of $F^{-1}$. This is the perspective argument.
Continuity extends it to $q=0$, and monotonicity in $q$ gives
convexity of $B_e(|q|)$.
\end{proof}

\subsection{An edge inequality and entropy production}

We now turn the scalar profile into a lower bound for the rate at
which information increases with correlation.

\begin{proposition}\label{prop:profile-production}
For every nonconstant $f:\{-1,1\}^n\to\{-1,1\}$ and $0<\rho<1$,
\[
 \rho A'_\rho(f)\ge\sum_i B_{E_\rho(f)}(\rho |\widehat f(i)|).
\]
\end{proposition}

The proof reduces to the following inequality on a single cube edge.

\begin{lemma}\label{lem:edge-entropy-profile}
For $x_-,x_+\in(-1,1)$, put $e=[H(x_+)+H(x_-)]/2$ and
$d=(x_+-x_-)/2$. Then
\[
 \frac{(x_+-x_-)(\atanh x_+-\atanh x_-)}4
 \ge B_e(|d|).
\]
\end{lemma}

\begin{proof}
Write $m=(x_++x_-)/2$. Changing the signs and interchanging the
endpoints allows us to assume $m,d\ge0$. The cases $d=0$ and
$m=0$ follow directly from the definition of $B$; hence suppose
$m,d>0$. Set $u=[\atanh(m+d)-\atanh(m-d)]/2$, $b=\tanh u$,
and $\ell=d/b$. The addition formula for $\tanh$ gives
$0<\ell\le1$ and $\ell+d^2/\ell=1-m^2+d^2$.
Keep $m$ fixed and differentiate with respect to $d$. The preceding
identity and $d=\ell b$ give
$\ell_d=-2b(1-\ell)/(1-b^2)$ and $e_d=-u$.
Also $(\ell H(b))_d=\ell_d\log(2\cosh u)-u$.
Thus $\partial_d(e-\ell H(b))\ge0$. As $d\downarrow0$, we have $\ell\to1-m^2$ and $b\to0$,
so $e-\ell H(b)\to H(m)-(1-m^2)\log2\ge0$ by
\eqref{ub:entropy-series}.
It follows that $e\ge\ell H(b)$. Since $B_e(d)$ decreases with $e$,
$du=B_{\ell H(b)}(d)\ge B_e(d)$,
which proves the claim.
\end{proof}

\begin{proof}[Proof of Proposition~\ref{prop:profile-production}]
Let $g=T_\rho f$. On a coordinate fiber, write $g_+,g_-$ for
its endpoint values, let $e_i$ be their average entropy, and set
$q_i=(g_+-g_-)/2$.
Fourier differentiation and summation over coordinate fibers give
\begin{align*}
 \rho A'_\rho(f)
 &=\sum_i\E_{x_{-i}}
    \left[q_i\frac{\atanh g_+-\atanh g_-}2\right]\\
 &\ge\sum_i\E_{x_{-i}} B_{e_i}(|q_i|).
\end{align*}
Here the last step is Lemma~\ref{lem:edge-entropy-profile}.
Since $\E e_i=E_\rho(f)$ and
$\E q_i=\rho\widehat f(i)$, joint convexity from
Lemma~\ref{u963:scalar} proves the proposition.
\end{proof}

\subsection{Integrating from the Boolean endpoint}

The production bound can be integrated because the original source
has zero conditional entropy at $\rho=1$. Use noise time
$\tau=-\log\rho$ and rescale entropy as $y=E_\rho(f)/\rho$.
The scaling identity for $B$ will give a differential inequality of
the form $dy/d\tau\ge y+P_f(y)$, where the function $P_f$ depends
only on the source's singleton coefficients. Separation of variables
then leads to the integral below, called a \emph{profile clock}:
it measures the time in the scalar comparison equation needed to
reach entropy level $y$.

Define the singleton profile and its clock by
\begin{equation}\label{eq:singleton-clock}
 P_f(y)=\sum_iB_y(|\widehat f(i)|),\qquad
 T_{P_f}(v)=\int_0^v\frac{dy}{y+P_f(y)}.
\end{equation}
The condition $W_1>0$ makes this integral finite near zero and holds
for every nonconstant increasing source. A dictator has one singleton
coefficient of magnitude one, so its profile is simply $B_y(1)$;
this gives the reference clock for the comparisons below.

\begin{theorem}[Profile clock]\label{thm:profile-clock}
If $f:\{-1,1\}^n\to\{-1,1\}$ and $W_1(f)>0$, then
for $0<\rho<1$,
\[
 T_{P_f}\bigl(E_\rho(f)/\rho\bigr)\ge-\log\rho.
\]
\end{theorem}

\begin{proof}
Set $\tau=-\log\rho$ and $y(\tau)=E_\rho(f)/\rho$.
By the scaling identity in Lemma~\ref{u963:scalar},
$\sum_i B_{E_\rho}(\rho |\widehat f(i)|)=\rho P_f(y)$.
Since $E'_\rho=-A'_\rho$, Proposition~\ref{prop:profile-production} yields
$dy/d\tau=y+A'_\rho(f)\ge y+P_f(y)$.
Hence $\frac{d}{d\tau}T_{P_f}(y(\tau))\ge1$ for $\tau>0$.
Booleanity and finite dimension give $y(\tau)\to0$ as
$\tau\downarrow0$. The clock is finite and continuous at zero:
for $0<y\le1$, $P_f(y)\ge P_f(1)>0$.
Integrating on $[\eta,\tau]$ and taking $\eta\downarrow0$ proves the clock inequality.
\end{proof}

More generally, let
\begin{equation}\label{eq:fixed-profile-envelope}
 Q(y)=\sum_{j=1}^J c_jB_y(q_j),
 \qquad J\ge1,\quad c_j>0,\quad 0<q_j\le1,
 \qquad T_Q(v)=\int_0^v\frac{dy}{y+Q(y)}.
\end{equation}
The weights $c_j$ and arguments $q_j$ are fixed as $\rho$ varies;
they may depend on the source cell under consideration.
The clock $T_Q$ is finite at zero and strictly increasing.
If $Q\le P_f$, then $T_{P_f}\le T_Q$.

\subsection{Propagation of an integrated entropy bound}

One clock comparison at the left endpoint will cover every larger
correlation. The envelope must stay fixed while the correlation varies.

\begin{theorem}[Propagation to larger correlations]\label{thm:clock-propagation}\label{st:clock}
Let $f:\{-1,1\}^n\to\{-1,1\}$ satisfy $W_1(f)>0$,
and let $Q$ have the fixed finite-sum form
\eqref{eq:fixed-profile-envelope}, with $Q\le P_f$ on $(0,\infty)$.
Fix $0<\rho_0<1$ and a constant $\epsilon\in[0,1)$, and suppose that
\[
 T_Q\bigl(H(\rho_0)/\rho_0\bigr)
 \le(1-\epsilon)\log(1/\rho_0).
\]
Then, for $\rho_0\le\rho<1$,
\[
 T_Q\bigl(H(\rho)/\rho\bigr)\le(1-\epsilon)\log(1/\rho),
 \qquad E_\rho(f)\ge H(\rho)\rho^{-\epsilon}.
\]
In particular, the choice $\epsilon=0$ gives
$E_\rho(f)\ge H(\rho)$ for $\rho_0\le\rho\le1$.
\end{theorem}

Here $\epsilon$ specifies a guaranteed fractional saving relative to
the dictator's clock value $\log(1/\rho)$; it remains fixed as $\rho$
varies. The value $\epsilon=0$ gives the non-strict bound, while
$\epsilon>0$ yields the factor $\rho^{-\epsilon}>1$ in the entropy
estimate. We first record the monotonicity of the individual profile
ratios.

\begin{lemma}\label{lem:profile-ratio}
For $0<a<1$, the function $u\mapsto F^{-1}(F(u)/a)/u$ is
strictly increasing on $(0,\infty)$. For $a=1$ it equals one.
\end{lemma}
\begin{proof}
Put $v=F^{-1}(F(u)/a)$, so $0<v<u$.
The function $\mathcal K$ in \eqref{eq:profile-elasticity} is strictly
increasing by the proof of Lemma~\ref{u963:scalar}.
Implicit differentiation gives
$d\log v/d\log u=\mathcal K(u)/\mathcal K(v)>1$.
Thus $v/u$ increases strictly with $u$.
\end{proof}

\begin{proof}[Proof of Theorem~\ref{thm:clock-propagation}]
Let $T_{\rm dic}$ be the clock of the dictator profile $B_y(1)$.
The substitution $y=H(t)/t$ gives
$dy/[y+B_y(1)]=-dt/t$, and therefore
\[
 T_{\rm dic}(H(\rho)/\rho)=-\log\rho.
\]
Compare the profile $Q$ with the dictator profile by putting
$R(y)=Q(y)/B_y(1)$ and $x(y)=y/B_y(1)$.
Lemma~\ref{lem:profile-ratio} shows that $R$ is nonincreasing.
Also $x$ is increasing: with $y=F(v)$ it equals $F(v)/v$,
which decreases with $v$. Consequently the derivative ratio
\[
 \frac{T'_Q(y)}{T'_{\rm dic}(y)}
 =\frac{x(y)+1}{x(y)+R(y)}
\]
is nondecreasing on the initial interval where $R\ge1$.
On the remaining interval it is at least one.

Both clocks vanish at zero. Their quotient $T_Q(y)/T_{\rm dic}(y)$
is thus the average of the displayed derivative ratio with respect
to $dT_{\rm dic}$ on $(0,y)$. This average is nondecreasing while
it is below one: first it averages an increasing function, and later
all new values are at least one. If the average ever exceeds one,
every later average remains strictly above one. The assumed
fractional bound at $H(\rho_0)/\rho_0$
therefore holds at every smaller positive $y$. Since $H(\rho)/\rho$
decreases with $\rho$, this proves the first assertion.

The profile-clock theorem and $Q\le P_f$ give
$T_Q(E_\rho(f)/\rho)\ge-\log\rho$, hence $E_\rho(f)\ge H(\rho)$.
Subtracting the clock bound just proved yields
\[
 \epsilon\log(1/\rho)
 \le\int_{H(\rho)/\rho}^{E_\rho(f)/\rho}\frac{dy}{y+Q(y)}
 \le\log\frac{E_\rho(f)}{H(\rho)}.
\]
Exponentiation proves the entropy estimate. At $\rho=1$, both
entropies in the non-strict conclusion vanish.
\end{proof}

\subsection{Profiles controlled by one or two coordinates}

To use the clock in a finite certificate, we replace the full
singleton profile by a lower bound depending on one or two retained
coefficients and a lower bound for $W_1$. Concavity shows that the
remaining squared mass may be concentrated at its allowed cap.

\begin{lemma}\label{jp:profile}
Let $f:\{-1,1\}^n\to\{-1,1\}$ be increasing.
Suppose $\widehat f(1)\in[a,A]$, where $0\le a\le A\le1$, and
$W_1\ge w$. Set
$c_1=\min(A,1-a)$. If $c_1>0$, write $(w-A^2)_+=k_1c_1^2+v_1$ with
$k_1\in\mathbb Z_{\ge0}$ and $0\le v_1<c_1^2$.
Then, for every $y>0$,
\[
 P_f(y)\ge B_y(a)+k_1B_y(c_1)+B_y(\sqrt{v_1}).
\]
If, in addition, $\widehat f(2)\in[s,S]$ with $0\le s\le S\le1$ and $S>0$, write
$(w-A^2-S^2)_+=k_2S^2+v_2$ with
$k_2\in\mathbb Z_{\ge0}$ and $0\le v_2<S^2$.
Then
\[
 P_f(y)\ge B_y(a)+B_y(s)+k_2B_y(S)+B_y(\sqrt{v_2}).
\]
Zero summands are omitted. If a tail cap is zero, its actual mass
is zero; for a nonempty source cell the corresponding lower mass
is then zero, and no division by the cap is needed.
\end{lemma}

\begin{proof}
Ordering and the pair bound \eqref{eq:pair-bound} give $\widehat f(i)\le c_1$
for $i\ge2$ and $\sum_{i\ge2}\widehat f(i)^2\ge(w-A^2)_+$.
Reduce tail masses to total $(w-A^2)_+$.
Since $v\mapsto B_y(\sqrt v)$ is concave and vanishes at zero,
any two interior masses can be moved to an endpoint, preserving
their sum and without increasing their contribution to the profile. After finitely many moves, $k_1$ masses
equal $c_1^2$ and at most one equals $v_1$.
Together with $B_y(\widehat f(1))\ge B_y(a)$, this proves the first bound.

For the second bound, keep the first two coefficients. Every
remaining coefficient is at most $\widehat f(2)\le S$, and their squared
mass is at least $(W_1-\widehat f(1)^2-\widehat f(2)^2)\ge(w-A^2-S^2)_+$. Apply the same
argument with cap $S^2$, then use $B_y(\widehat f(2))\ge B_y(s)$.
\end{proof}

This redistribution relaxes the Fourier constraints; its concentrated
profile need not be realizable by a Boolean source. The resulting
bound depends only on the cell parameters.

The same concavity also reduces the number of source values
that need checking. Recall that a \emph{polytope} is the convex
hull of finitely many points in Euclidean space. Every point
is a convex combination of its vertices, so a convex function
on a polytope attains its maximum at a vertex. We apply this
observation when the remaining squared masses depend affinely
on the retained squared coefficients.

\begin{lemma}[Convexity in squared influences]\label{lem:source-convexity}
Let $K$ be a polytope and let $\ell_1,\ldots,\ell_m:K\to[0,\infty)$
be affine functions with $\sum_j\ell_j>0$ on $K$. For positive weights
$c_1,\ldots,c_m$, put $Q_x(y)=\sum_jc_jB_y(\sqrt{\ell_j(x)})$.
For every $Y>0$, the function
\[
 x\longmapsto\int_0^Y\frac{dy}{y+Q_x(y)}
\]
is convex on $K$ and attains its maximum at a vertex.
\end{lemma}
\begin{proof}
Lemma~\ref{u963:scalar} makes $Q_x(y)$ concave in $x$.
Since the reciprocal is convex and decreasing on $(0,\infty)$,
$1/(y+Q_x(y))$ is convex in $x$. Integration preserves convexity;
the integral is finite because $Q_x(y)\ge Q_x(Y)>0$ for $0<y\le Y$.
If $x=\sum_r\theta_r v_r$ is a convex combination of vertices,
convexity bounds the integral at $x$ by the corresponding weighted
average at the $v_r$, hence by their largest value. This proves
the last assertion.
\end{proof}

The next observation allows the cap itself to vary. It removes
artificial subdivisions when a source constraint gives an exact cap.

\begin{lemma}[Concavity with a varying cap]\label{lem:moving-cap}
Let $b:[0,\infty)\to\R$ be continuous, concave, and twice
continuously differentiable on $(0,\infty)$. On an interval $I$,
let $c:I\to(0,\infty)$ be convex and twice continuously
differentiable, and let $w:I\to\R$ be affine. If $k\ge0$ and
$0\le w-kc\le c$ on $I$, then $kb(c)+b(w-kc)$ is concave.
\end{lemma}
\begin{proof}
Write $v=w-kc$. Where $v>0$, the second derivative is
\[
 kb''(c)(c')^2+b''(v)(v')^2
       +kc''\bigl(b'(c)-b'(v)\bigr)\le0.
\]
Here $v\le c$ and $b'$ is decreasing. Continuity handles zero
residuals at the endpoints. If $v$ vanishes at an interior point,
its concavity and nonnegativity force $v$ to vanish identically;
then $kc$ is affine and the assertion follows directly.
\end{proof}

\section{Entropy and energy at intermediate correlations}
\label{sec:cover}

This section covers $457/500\le\rho\le49/50$.
The local bounds from Section~\ref{ub:lower} handle functions close
to one or two coordinates; a fourth-moment estimate pays the mean
correction uniformly. For the remaining sources, we split
according to the first-level energy $W_1$. If it is small, Harris's
inequality controls $W_2$ and a cubic energy bound proves CK.
If it is large, the integrated entropy bound proves the stronger
inequality $E_\rho(f)\ge H(\rho)$. The certificate assigns each
source cell to one of these arguments and checks that the cells
cover every source.

\subsection{Conditional Harris inequalities without balance}

Harris association bounds $W_2$ in terms of $W_1$ and one or two
retained coefficients. These are the source constraints needed by
the cubic energy bound in the next subsection.

\begin{theorem}[Harris inequality {\cite{Harris1960}}]
Let $X_1,\ldots,X_n$ be independent signs, with arbitrary individual
biases. For increasing functions $f,g:\{-1,1\}^n\to\R$,
$\Cov(f(X),g(X))\ge0$.
\end{theorem}

We apply this theorem to the monotone source obtained by compression.
It remains valid after conditioning on any coordinates, because the
unconditioned coordinates still have a product law.

Pairing $f$ with its noisy dual alternates the signs of the Fourier
levels. Conditional Harris association turns that alternation into an
upper bound on $W_2$.

The dual $f^{\rm d}(x)=-f(-x)$ is increasing and satisfies
$\widehat {f^{\rm d}}(S)=(-1)^{|S|+1}\widehat f(S)$, including
the constant coefficient $-\mu$. Conditional association of
$f$ and $T_\eta f^{\rm d}$ yields
\begin{equation}\label{jp:core-harris}
 \sum_{S\ne\varnothing}(-1)^{|S|+1}\eta^{|S|}\widehat f(S)^2
 \ge\sum_{\varnothing\ne S\subseteq K}
       (-1)^{|S|+1}\eta^{|S|}\widehat f(S)^2.
\end{equation}
Indeed, the difference is the mean of their conditional
covariances outside $K$. The global covariance and that of the
conditional means both remove the constant term $-\mu^2$.

\begin{proposition}
Let $f:\{-1,1\}^n\to\{-1,1\}$ be increasing. For $0<\eta\le1$, we have
\begin{align*}
 \eta(1+\eta)W_2
 &\le(1-\eta^2)W_1+\eta^2(1-\mu^2)-\widehat f(1)^2,
 \\
 \eta(1+\eta)W_2
 &\le(1-\eta^2)W_1+\eta^2(1-\mu^2)-\widehat f(1)^2-\widehat f(2)^2+\eta \widehat f(\{1,2\})^2
 \\
 &\le(1-\eta^2)W_1+\eta^2-\widehat f(1)^2-(1-\eta)\widehat f(2)^2.
\end{align*}
\end{proposition}
\begin{proof}
For $k\ge3$, $(-1)^{k+1}\eta^k\le\eta^3$. Thus the left
side of \eqref{jp:core-harris} is at most
$\eta W_1-\eta^2W_2+\eta^3(1-\mu^2-W_1-W_2)$.
Take $K=\{1\}$ or $K=\{1,2\}$, divide by $\eta$, and rearrange.
These choices give the first two inequalities directly. For the last,
use $|\widehat f(\{1,2\})|\le\widehat f(2)$ and discard
$-\eta^2\mu^2$.
\end{proof}

The interaction constraint is
$|\mu+\widehat f(\{1,2\})|\le1-\widehat f(1)-\widehat f(2)$;
both Harris bounds hold simultaneously for every feasible mean.

\subsection{Cubic energy and the prior-entropy credit}

We now give the low-energy half of the dichotomy. The relevant
quantity is
\begin{equation}\label{eq:cubic-weight}
 \mathcal W_\rho=W_1+\frac{\rho^2}{1+\rho}W_2.
\end{equation}
Taking $d=3$ in Lemma~\ref{ub:degree-energy} and factoring $1-\rho$ gives
\begin{equation}\label{ub:cubic-product}
 \begin{split}
 \E[(1-|g|)(1+|g|-\rho^3)]\ge(1-\rho)
 \{(1+\rho+\rho^2)(1-\mu^2)-\rho^2(1+\rho)\mathcal W_\rho\}.
 \end{split}
\end{equation}

For a cap $M$ define
\begin{equation}\label{eq:cubic-parameters}
 C_\rho=1+\rho+\rho^2-\rho^2(1+\rho)M,
 \qquad\kappa_\rho=\frac{H(\rho)}{(1-\rho)C_\rho}.
\end{equation}
Suppose $C_\rho>0$ and the scalar conditions
\begin{equation}\label{ub:cubic-gates}
 \begin{gathered}
 H(t)\ge\kappa_\rho(1-t)(1+t-\rho^3)\quad(0\le t\le1),\\
 \kappa_\rho(1-\rho^3)<1/2
 \end{gathered}
\end{equation}
hold. On the branch $\mathcal W_\rho\le M$,
Equation~\eqref{ub:cubic-product} gives
\begin{equation}\label{ub:cubic-bias}
 E_\rho(f)-H(\rho)+\Phi(\mu)
 \ge\Phi(\mu)-\kappa_\rho(1-\rho^3)\mu^2\ge0.
\end{equation}
The inequality for the posterior mean in \eqref{ub:cubic-gates} is checked
by Lemma~\ref{ub:tangent}, with $c=\kappa_\rho$ and $b=\rho^3$.

Convexity in the correlation reduces each scalar inequality to the two
endpoints of its band.

\begin{lemma}[Endpoint propagation]\label{ub:cubic-endpoints}
Let $3/4\le M\le83/100$ and $457/500\le\rho_-<\rho_+<1$.
Let $C_\rho$ and $\kappa_\rho$ be as in \eqref{eq:cubic-parameters}. If
$H(t)\ge\kappa_\rho(1-t)(1+t-\rho^3)$ for every $0\le t\le1$
at both $\rho=\rho_-$ and $\rho=\rho_+$, then the same inequality
holds for every $\rho\in[\rho_-,\rho_+]$.
Moreover, $C_\rho>0$ and
$\kappa_\rho(1-\rho^3)<30/67<1/2$ throughout $457/500\le\rho<1$.
\end{lemma}
\begin{proof}
The proof uses convexity in the correlation above $t=1/2$ and
monotonicity in $t$ below that point. We first verify the curvature.

Put $x=1-\rho$, so $0<x\le43/500$, and write
\[
 \frac{H(\rho)}{1-\rho}
 =\frac12\log\frac{2e}{x}-R(x),\qquad
 R(x)=\sum_{k\ge1}\frac{x^k}{2^{k+1}k(k+1)},\qquad
 V(x)=\frac{1/2+3x-3x^2+x^3}{C_\rho}.
\]
Thus $\kappa_\rho(3/2-\rho^3)
=[\frac12\log(2e/x)-R(x)]V(x)$.
Lemma~\ref{app:cubic-curvature} gives
$-7+30x<V''(x)<0$ and $V'(x)<2$.
Also $V(0)=1/[2(3-2M)]\ge1/3$ and $V(x)<3/5$.
The latter follows from $C_\rho\ge67/50$ and
$1/2+3x-3x^2+x^3\le1/2+3x$.
The series gives
$0<R'(x)\le1/8+x/[24(1-x/2)]<13/100$ and
$0<R''(x)\le1/[24(1-x/2)^2]<1/20$.
By concavity, $V(x)-2xV'(x)\ge V(0)-xV'(x)\ge1/3-2x$.
Differentiating twice and discarding the positive term $-R(x)V''(x)$
therefore gives
\[
 2x^2\frac{d^2}{d\rho^2}
       [\kappa_\rho(3/2-\rho^3)]
 >\frac13-2x-x^2\left[(7-30x)\log\frac{2e}{x}
                                      +\frac{11}{10}\right]>0.
\]
For the last inequality, the right side decreases on $(0,43/500]$:
the derivative of $x^2(7-30x)\log(2e/x)$ is
$x[(14-90x)\log(2e/x)-7+30x]>0$ because $\log(2e/x)>4$.
At $x=43/500$, the bound $\log(2e/x)<21/5$ leaves
$745237/46875000>0$. Both logarithm bounds follow from
$20<e^3<21$; for the upper bound also use $e^{1/5}>6/5$.

As a function of $\rho$, $C_\rho$ is positive, decreasing, and concave
on the stated range. The function $H(\rho)/(1-\rho)$ is increasing
and strictly convex: its first derivative is
$\log(2/(1+\rho))/(1-\rho)^2$, and its second derivative is at least
$\rho/[(1-\rho)^2(1+\rho)]>0$.
Consequently $\kappa_\rho''>0$. The curvature just proved then shows
that $\rho\mapsto\kappa_\rho(1+t-\rho^3)$ is strictly convex for
every $t\ge1/2$.

Let $\chi_\rho(t)=-2H(t)/(1-t)+2\kappa_\rho(1+t-\rho^3)$.
For $t\ge1/2$, convexity in $\rho$ and the endpoint hypotheses give
$\chi_\rho(t)\le0$. For $t\le1/2$, concavity in $t$ from
Lemma~\ref{ub:tangent} gives
$\chi_\rho'(t)\ge\chi_\rho'(1/2)=8\log(3/4)+2\kappa_\rho>0$.
Indeed, $C_\rho<8/5$, $R(x)<1/16$, and
$\log(2e/x)>4$ give $\kappa_\rho>155/128>6/5$, whereas
$\log(4/3)<3/10$. Hence
$\chi_\rho(t)\le\chi_\rho(1/2)\le0$ there as well.
Multiplying by $(1-t)/2$ proves the inequality for $t<1$;
the case $t=1$ follows by continuity.

Finally, $H(\rho)<H(9/10)<1/5$ and $C_\rho\ge67/50$, so
$\kappa_\rho(1-\rho^3)=H(\rho)(1+\rho+\rho^2)/C_\rho<30/67$.
To see $H(9/10)<1/5$, use $\log20<3$ and
$\log(20/19)<1/19$ in the binary entropy formula.
\end{proof}

\subsection{Combining the entropy and energy bounds}

In the scalar comparisons below, $a$ and $b$ range over the possible
values of $\widehat f(1)$ and $\widehat f(2)$, respectively;
$c$ denotes the possible value of $\widehat f(\{1,2\})$.
We choose a first-level threshold so that the energy bound covers
one side and a fixed profile clock covers the other. This is the
criterion checked on each source cell. Here a \emph{source cell} consists of increasing Boolean
functions whose one or two largest singleton coefficients
lie in prescribed intervals. We use the singleton profile $P_f$
from \eqref{eq:singleton-clock} and the envelope clock $T_Q$
from \eqref{eq:fixed-profile-envelope}.

Fix a band $[\rho_-,\rho_+]$ and put
$\beta_+=\rho_+^2/(1+\rho_+)$. For a source cell
$a\in[a_0,A]$, $b\in[b_0,B]$, choose rational $w,\eta$ with
$0<\eta\le1$ such that
\begin{equation}\label{ub:threshold}
 w+\beta_+
 \frac{(1-\eta^2)w+\eta^2-a_0^2-(1-\eta)b_0^2}
      {\eta(1+\eta)}\le M.
\end{equation}
For a one-coordinate cell omit $b_0$. The right side of the
Harris bound increases with $W_1$ and decreases with $a,b$.
Also $\mathcal W_\rho\le W_1+\beta_+W_2$.
Consequently $W_1\le w$ implies $\mathcal W_\rho\le M$ on the
whole band.
If the resulting upper bound on $W_2$ is negative, that energy
branch has no feasible source. The separate record $\eta=0$
means the Parseval test
\begin{equation}\label{ub:parseval-threshold}
 \beta_++(1-\beta_+)w\le M;
\end{equation}
it is not a substitution into \eqref{ub:threshold}.

On the complementary branch $W_1\ge w$, use one of the fixed
profile envelopes in Lemma~\ref{jp:profile}. If its clock satisfies
\begin{equation}\label{ub:clock-trigger}
 T_Q(H(\rho_-)/\rho_-)<-\log\rho_-,
\end{equation}
Theorem~\ref{thm:clock-propagation} proves the stronger
$E_\rho(f)\ge H(\rho)$ for all $\rho\ge\rho_-$. It therefore
proves CK for every mean on that branch. At $W_1=w$ both
alternatives are valid.

The threshold condition gives CK when $W_1\le w$, while the
clock condition covers $W_1\ge w$. The following lemma combines
these alternatives into a criterion for an entire source cell.

\begin{lemma}\label{ub:cell-dichotomy}
Fix a correlation band, a source cell, and an energy cap $M$.
Suppose $C_\rho>0$ and \eqref{ub:cubic-gates} hold throughout the band. Assume that for each source in the
cell there is a threshold $w$ such that
\begin{enumerate}[label=(\alph*)]
 \item $W_1\le w$ forces $\mathcal W_\rho\le M$ by
       \eqref{ub:threshold} or \eqref{ub:parseval-threshold}; and
 \item $W_1\ge w$ forces an envelope $Q\le P_f$ of the fixed form
       \eqref{eq:fixed-profile-envelope}, independent of $\rho$,
       whose clock satisfies \eqref{ub:clock-trigger}.
\end{enumerate}
Then CK holds for every source in the cell throughout the band.
\end{lemma}

\begin{proof}
The first alternative gives \eqref{ub:cubic-bias}.  In the second,
Theorem~\ref{thm:clock-propagation} gives the stronger conclusion
$E_\rho(f)\ge H(\rho)$.  The alternatives meet at $W_1=w$.
\end{proof}

The entropy profile $B_e$ in \eqref{eq:entropy-profile} is convex
in $e$ by Lemma~\ref{u963:scalar}. Its supporting lines give explicit
bounds for the clock integral without evaluating $F^{-1}$.

\begin{lemma}[Supporting-line integration]\label{lem:clock-tangent}
For every $e,q,v>0$,
\[
 B_e(q)\ge qv+(qF(v)-e)\frac{\sinh^2v}{\log(2\cosh v)}.
\]
Let $Q(e)=\sum_{j=1}^m c_jB_e(q_j)$, with $c_j,q_j>0$.
For fixed $v_j>0$, let $L_j(e)$ be the right-hand side above with
$(q,v)=(q_j,v_j)$, and set $\ell(e)=e+\sum_jc_jL_j(e)$.
On any interval $0\le a<b$ where
$\ell(a),\ell(b)>0$, the integral
$\int_a^b de/(e+Q(e))$ is at most
$\log(\ell(b)/\ell(a))/\ell'$.
If $\ell'=0$, the bound is $(b-a)/\ell(a)$.
The integral at $a=0$ is improper.
\end{lemma}
\begin{proof}
For fixed $q>0$, differentiating \eqref{eq:entropy-profile} at
$e_0=qF(v)$ gives
\[
 B_{e_0}(q)=qv,\qquad
 \left.\frac{d}{de}B_e(q)\right|_{e=e_0}
 =\frac1{F'(v)}=-\frac{\sinh^2v}{\log(2\cosh v)}.
\]
By Lemma~\ref{u963:scalar}, $e\mapsto B_e(q)$ is convex, so its
tangent at $e_0$ lies below it for every $e>0$. This is the first inequality.

Summing with the same weights gives $e+Q(e)\ge\ell(e)$.
The function $\ell$ is affine, so positivity at both endpoints
implies positivity throughout $[a,b]$. Hence
$1/(e+Q(e))\le1/\ell(e)$ for $e>0$ in this interval.
Writing $s=\ell'$ and integrating $1/[\ell(a)+s(e-a)]$ gives
$\log(\ell(b)/\ell(a))/s$ when $s\ne0$, and
$(b-a)/\ell(a)$ when $s=0$.
If $a=0$, first integrate from $\varepsilon>0$ to $b$ and let
$\varepsilon\downarrow0$. The positive affine lower bound makes
the limiting upper integral finite, justifying the improper integral.
\end{proof}

Choose rational $Y>H(\rho_-)/\rho_-$ and partition $[0,Y]$.
On each interval choose any rational $t\in(0,1)$ for each profile
term and apply Lemma~\ref{lem:clock-tangent} with $v=\atanh t$. These choices need not
solve any equation: convexity certifies every supporting line.
The verifier checks positivity at both interval endpoints and sums
the explicit integrals. This bounds $T_Q(Y)$, including the
interval adjacent to zero. All source-envelope data remain fixed
as the noise varies.

\subsection{Four one-coordinate bands}

The dichotomy covers $457/500\le\rho\le39/40$ using one retained
coefficient. We keep its exact value in the energy threshold and
profile, then use convexity to check only finitely many boundary
values of its square.

\begin{theorem}\label{ub:upper-middle}
Theorem~\ref{jp:main} holds on $457/500\le\rho\le39/40$.
\end{theorem}
\begin{proof}[Computer-assisted proof]
The following finite decimals are exact rational parameters.
\begin{center}\small
\begin{tabular}{@{}cccc@{}}
\toprule
Band & $M$ & $A_*$ & Clock checks\\
\midrule
$[.914,.95]$ & $.82523844$ & $.66996$ & 7\\
$[.95,.97]$ & $.787706016$ & $.72245$ & 5\\
$[.97,.974]$ & $.77592593$ & $.73625$ & 4\\
$[.974,.975]$ & $.772609873$ & $.73996$ & 4\\
\bottomrule
\end{tabular}
\end{center}
For each band set $\beta_+=\rho_+^2/(1+\rho_+)$ and choose the threshold
\[
 w(a)=\max\left\{\frac{M-\beta_+}{1-\beta_+},\,
                  M-\frac{\beta_+}{2}(1-a^2)\right\}.
\]
If $W_1\le w(a)$, either Parseval or \eqref{ub:threshold} with
$\eta=1$ gives $\mathcal W_\rho\le M$ throughout the band.
This single rule replaces individually chosen thresholds and Harris
parameters.

For $W_1\ge w(a)$ and $a\le2/5$, concavity in squared
influences gives a single envelope for the whole region:
$P_f(y)\ge(25w(0)/4)B_y(2/5)$.
Indeed, every squared coefficient is at most $4/25$, and
$B_y(\sqrt s)\ge(25s/4)B_y(2/5)$ there. Sum and use
$W_1\ge w(a)\ge w(0)$.

For $2/5\le a\le A_*$, retain the exact cap $C=\min(a,1-a)$
on the other coefficients. Write $(w(a)-a^2)_+=kC^2+v$, where
$0\le v<C^2$. Packing gives
$Q_a(y)=B_y(a)+kB_y(C)+B_y(\sqrt v)$.
Split only at $a=1/2$, the change of threshold formula, and
packing transitions $w(a)-a^2=kC^2$.
On each region, $w(a)-a^2$ is affine in $a^2$.
The squared cap is affine in $a^2$ below $1/2$; above it,
$(1-a)^2$ is convex as a function of $a^2$.
Lemma~\ref{lem:moving-cap}, applied to $b(s)=B_y(\sqrt s)$,
shows that the packed profile is concave in $a^2$.
The reciprocal integral is therefore convex, as in
Lemma~\ref{lem:source-convexity}, so its maximum occurs at a
boundary value. These boundaries solve explicit linear or quadratic
equations in $a$ and are reconstructed exactly. The initial envelope
also covers the boundary $a=2/5$.
The resulting 20 clocks pass \eqref{ub:clock-trigger}, using
36 integration intervals and 91 supporting-line evaluations.

For $a\ge A_*$, each right endpoint verifies $M_\rho(A_*)>0$
for the local deficit in \eqref{eq:local-deficit};
Lemma~\ref{ub:local-propagation} and Theorem~\ref{ub:local} cover
the band. Lemma~\ref{ub:cubic-endpoints} reduces the scalar gates to
eight correlation endpoint checks and supplies the prior credit
analytically. The replay checks those gates, the four local signs,
and exact adjacency of the source and correlation intervals.
Lemma~\ref{ub:cell-dichotomy} now covers every source.
\end{proof}

\subsection{The two-coordinate band}

In the last intermediate band, convexity still handles $a\le2/3$
with one retained coordinate. Above that cutoff we use both largest
singleton coefficients, through either a profile comparison or the
exact two-coordinate entropy bound.

\begin{theorem}\label{ub:last-middle}
Theorem~\ref{jp:main} holds on $39/40\le\rho\le49/50$.
\end{theorem}
\begin{proof}[Computer-assisted proof]
Set $A_*=76043/100000$ and $M=150614067/200000000$.
Use the preceding threshold and the same cap rule on
$0\le a\le2/3$. Four clock comparisons suffice, using nine
integration intervals and 21 supporting-line evaluations.
The one-coordinate local sign $M_\rho(A_*)>0$ covers $a\ge A_*$.
It remains to cover $2/3<a<A_*$. Restrict the certified rational
partition of $[0,A_*]\times[0,1/2]$ to this region.
Before each split, a rectangle $[a_0,A]\times[b_0,B]$ is contracted
as follows: first replace $a_0$ by $\max(a_0,b_0)$, then $A$
by $\min(A,1-b_0)$, and finally $B$ by $\min(B,A,1-a_0)$.
Every feasible point with $a\ge b$ and $a+b\le1$ survives this
operation. Each binary split retains both closed children.

We first cover $3a+7b\ge64/25$ by a uniform moment bound. Apply Lemma~\ref{lem:mean-moments} to the noisy conditional
mean from Proposition~\ref{jp:local2-exact}, with
$\lambda=1-\rho^2<1/20$. Its centered third moment is
$6\rho^4abc$. Since $|c|\le b\le1-a$ and $a\ge2/3$,
its absolute value is at most $6a(1-a)^2\le4/9$.
Thus the compensation error is less than $1/5000$.
The second and fourth centered moments increase with $|c|$, so we
may replace it by $b$ in their negative terms.

For these moment bounds define $h:\{-1,1\}^2\to\R$ by
$h(x,y)=\rho ax+\rho by+\rho^2bxy$, with uniform expectation.
Lemma~\ref{lem:and-rational}
and $9/13<\log2<7/10$ give
\begin{equation}\label{eq:local-moment-polynomial}
 \frac{\Phi(\rho)-A_\rho(f)}{1-\rho^2}
 \ge\frac{17}{12}[\rho^2(a+b)-1]-\frac{\rho^2b}{5}
       +\frac9{13}-\frac{\E h^2}{2}-\frac{\E h^4}{5}-\frac1{198}.
\end{equation}
Here $1/[5000(1-\rho^2)]\le1/198$.
For fixed $\rho$, the right side is concave in $(a,b)$.
Its minimum on $2/3\le a\le A_*$, $a+b\le1$, and $3a+7b\ge64/25$ therefore occurs at one of the four vertices
$(2/3,2/25)$, $(2/3,1/3)$, $(A_*,1-A_*)$, and $(A_*,(64/25-3A_*)/7)$.
For fixed $a,b$, both moments are polynomials in $\rho^2$ with
nonnegative coefficients, so the right side is also concave in
$\rho^2$. Evaluation at the two correlation endpoints gives eight
rational values, all greater than $1/1000$.
This proves the local claim without a mean partition.

Exactly 22 certified leaves meet $a>2/3$. Seventeen satisfy
$3\max(a_0,2/3)+7b_0\ge64/25$ and are covered by the moment bound.
The remaining five use Lemma~\ref{ub:cell-dichotomy}, with the
two-coordinate threshold \eqref{ub:threshold} and a one- or
two-coordinate profile. Their clocks use 11 integration intervals
and 40 supporting-line evaluations. Lemma~\ref{ub:cubic-endpoints}
reduces the scalar conditions to two correlation endpoint checks,
and the one-coordinate local sign propagates from the right endpoint.

The replay reconstructs the full source tree, verifies the
classification of every retained leaf, and checks all thresholds,
envelopes, and integrals. The script \path{two_core_moment_verify.py}
checks the elementary constants and the eight rational values in
\eqref{eq:local-moment-polynomial}. A separate exact audit expands
the moments and checks the polynomial on the whole squared-correlation
interval using Bernstein coefficients. The one-coordinate argument
covers the omitted region, including $a=2/3$.
Compression completes the assertion for arbitrary sources.
\end{proof}

\section{Entropy interpolation and spectral transfer}
\label{sec:full-head}

By \eqref{eq:gap-derivative}, $\rho G_f'(\rho)=D(T_\rho f)-\rho\atanh\rho$.
Thus, at
a hypothetical positive gap, it is enough to rule out
$D(T_\rho f)\le\rho u$, where $u=\atanh\rho$.  We do this by
transforming $g=T_\rho f$ to $h=\arcsin g$. The choice reflects
entropy curvature: $(\arcsin)'(t)=\sqrt{\Phi''(t)}$, so
Cauchy--Schwarz gives
$(\arcsin b-\arcsin a)^2\le(b-a)(\atanh b-\atanh a)$ for $a<b$.
On each cube edge,
entropy production exceeds the Dirichlet energy of $h$ by a
nonnegative remainder. Globally, $h$ is close in $L_2$ to
$(\pi/2)g$, so spectral completion exposes the Fourier energy of the
noisy source.  The profile budget converts the remaining coordinate
sum into a scalar inequality.

Throughout this section $0<\rho<1$ and
$f:\{-1,1\}^n\to\{-1,1\}$ is nonconstant.
Write $u=\atanh\rho$, $g=T_\rho f:\{-1,1\}^n\to(-1,1)$, and
$h=\arcsin g:\{-1,1\}^n\to(-\pi/2,\pi/2)$.
Define the scalar entropy correction and its average by
\begin{equation}\label{eq:entropy-correction}
 \Psi(t)=t^2\log2-\Phi(t)=H(t)-(1-t^2)\log2,
 \qquad \mathcal J=\E\Psi(g).
\end{equation}
Thus $\mathcal J$ measures
the excess of conditional entropy above its elementary quadratic
lower bound. It is nonnegative by the entropy series and vanishes
at the scalar values $0,\pm1$. The next lemma uses this same excess
to control the error in replacing $\arcsin g$ by $(\pi/2)g$.
All expectations use the uniform probability measure. By \eqref{ub:gap-definitions},
$G_f(\rho)>0$ implies
$\E H(g)<H(\rho)-\Phi(\mu)\le H(\rho)$.

\subsection{A sharp scalar error estimate}

The arcsine transform is close to $\frac{\pi}{2}g$ at a cost measured by the
entropy remainder $\mathcal J$. This bound lets the later spectral
argument work directly with the Fourier coefficients of $f$.

\begin{lemma}\label{fh:arcsine-error}
For every $-1\le t\le1$,
\begin{equation}\label{eq:arcsine-constant}
 (\arcsin t-\frac{\pi}{2}t)^2\le C_*\Psi(t),
 \qquad C_*:=\frac{(\frac{\pi}{2}-1)^2}{{\log2}-1/2}.
\end{equation}
In particular,
\[
 \|h-\frac{\pi}{2}g\|_2^2\le C_*\mathcal J.
\]
\end{lemma}

\begin{proof}
By \eqref{ub:entropy-series}, $\Psi\ge0$ on $[-1,1]$.
We compare the two sides by their even power series. Put
$a_n=4^{-n}\binom{2n}{n}$ and write
\[
 U(t)=C_*\Psi(t)-(\arcsin t-\tfrac\pi2t)^2
     =\sum_{k\ge1}b_kt^{2k}.
\]
The choice of $C_*$ makes $b_1=0$. The standard series for
$\arcsin t$ and its square give, for $n\ge1$,
\begin{equation}\label{fh:error-coefficients}
 b_{n+1}=\frac{c_n-C_*}{2(n+1)(2n+1)},
 \qquad c_n=2\pi(n+1)a_n-\frac2{a_n}.
\end{equation}
Thus $c_n$ is the threshold determining the sign of the coefficient.
These thresholds decrease with $n$. Indeed, the Wallis integrals give
\[
 \frac\pi2a_n=\int_0^{\pi/2}\sin^{2n}x\,dx
 <\int_0^{\pi/2}\sin^{2n-1}x\,dx=\frac1{2na_n},
\]
so $\pi n a_n^2<1$. Using
$a_{n+1}/a_n=(2n+1)/(2n+2)$, we obtain
\[
 c_{n+1}-c_n
 =\frac{\pi n(2n+1)a_n^2-2(n+1)}{(n+1)(2n+1)a_n}<0.
\]

It remains to locate the single sign change. We have
\[
 c_3=\frac{5\pi}{2}-\frac{32}{5}
 <C_*<
 \frac{9\pi}{4}-\frac{16}{3}=c_2.
\]
For the logarithm bounds, retain the first two terms in
$\log2=2\sum_{j\ge0}((2j+1)3^{2j+1})^{-1}$ and bound the
remaining denominators below by $5$.
Consequently $b_2,b_3>0$ and $b_k<0$ for every $k\ge4$.

The series converge absolutely at $t=1$, where $U(1)=0$.
This endpoint identity and the single sign change determine the
sign on the whole interval: for $x=t^2\in[0,1]$, each negative
term satisfies $b_kx^k\ge b_kx^3$ when $k\ge4$. Hence
\[
 U(t)\ge b_2x^2+
       \left(b_3+\sum_{k\ge4}b_k\right)x^3
       =b_2x^2(1-x)\ge0.
\]
Evenness proves the scalar inequality on $[-1,1]$.
The ratio of its left-hand side to $\Psi(t)$ tends to $C_*$ as
$t\to0$, proving optimality. Averaging proves the function inequality.
\end{proof}

\subsection{An exact edge refinement}

We now bound the excess of entropy production over arcsine energy.
The main result below combines these edge remainders over any chosen
set of coordinates and gives a common budget for their profile heights.

For $v>0$, define
\begin{equation}\label{eq:arcsine-remainder}
 \zeta(v)=\arcsin(\tanh v)=\arctan(\sinh v),\qquad
 R(v)=v-\frac{\zeta(v)^2}{\tanh v},
\end{equation}
with $R(0)=0$. On a centered edge with endpoint magnitude $\tanh v$,
the production is $v\tanh v$ and the arcsine energy is $\zeta(v)^2$.
Thus $R(v)$ is the excess production per unit of edge size.
To evaluate this excess at the dictator entropy, use $F$ from
\eqref{eq:entropy-profile} and put
\begin{equation}\label{eq:profile-height}
 \ell_u(a)=F^{-1}(F(u)/a)\quad(a>0),\qquad \ell_u(0)=0.
\end{equation}
Indeed, $H(\rho)/(\rho a)=F(u)/a$, so $\ell_u(a)$ is precisely
the profile parameter associated with entropy $H(\rho)$ and
coordinate size $\rho a$. We call it a profile height. The next
proposition bounds the sum of the corresponding coordinate costs
by the total entropy production, which is the common budget used
in the spectral argument.

\begin{proposition}[Simultaneous arcsine refinement]
\label{fh:simultaneous-refund}
Let $f:\{-1,1\}^n\to\{-1,1\}$ be nonconstant,
$0<\rho<1$, $g=T_\rho f$, and $h=\arcsin g$.
For every $A\subseteq[n]$,
\[
 {D(g)}\ge\operatorname{Dir}(h)
 +\rho\sum_{i\in A\colon |\widehat f(i)|>0}
     |\widehat f(i)|R\left(F^{-1}\left(\frac{E_\rho(f)}{\rho |\widehat f(i)|}\right)\right).
\]
Put $u=\atanh\rho$ and $\ell_i=\ell_u(|\widehat f(i)|)$,
using \eqref{eq:profile-height}. If $G_f(\rho)>0$, then
\[
 {D(g)}\ge\operatorname{Dir}(h)
       +\rho\sum_{i\in A}|\widehat f(i)|R(\ell_i).
\]
If $G_f(\rho)>0$, the common profile budget satisfies
\[
 \sum_i |\widehat f(i)|\ell_i\le D(g)/\rho.
\]
\end{proposition}

We prove this by comparing each edge with a centered edge and then
averaging. The following convexity allows the edge entropies to be
replaced by their common mean.

\begin{lemma}\label{fh:refund-convexity}
The function $R$ is strictly increasing and strictly convex on
$(0,\infty)$. Define
\begin{equation}\label{eq:edge-remainder}
 \mathcal R(e,d)=dR(F^{-1}(e/d))\qquad(e,d>0).
\end{equation}
This function is jointly convex, increasing in $d$, and decreasing in $e$.
It extends continuously to $d=0$ with value zero, and
$(e,q)\mapsto\mathcal R(e,|q|)$ is jointly convex.
\end{lemma}

\begin{proof}
Direct differentiation gives
$R'(v)=[1-\arctan(\sinh v)/\sinh v]^2$.
For $s>0$, $\arctan s/s$ strictly decreases from $1$ to $0$,
because $\arctan s>s/(1+s^2)$.  Thus $R',R''>0$.
By Lemma~\ref{as:profile-convex}, the inverse of $F$ is convex and
decreasing. Composition with the increasing convex function $R$
shows that $R\circ F^{-1}$ is convex and decreasing.
The perspective argument in the proof of Lemma~\ref{u963:scalar}
therefore gives joint convexity of $\mathcal R$; its decrease in $e$
follows directly from the defining formula.
Increasing $d$ increases both nonnegative factors in $\mathcal R$.
Also $F^{-1}(e/d)\to0$ as $d\downarrow0$, which proves the
continuous extension. Monotonicity in $d$ and the triangle inequality
give convexity of the even extension.
\end{proof}

\begin{lemma}\label{fh:edge-inequality}
Let $x=c+d$, $y=c-d$ belong to $(-1,1)$, with $d>0$, and put
\[
 b=\frac{\atanh x-\atanh y}{2},\qquad
 A=\frac{\arcsin x-\arcsin y}{2},\qquad
 e=\frac{H(x)+H(y)}2.
\]
Then
\[
 db-A^2\ge dR(F^{-1}(e/d)).
\]
Equality holds when $c=0$, and the inequality is strict when
$c\ne0$.  A zero-length edge contributes zero to both sides, with
the extension in Lemma~\ref{fh:refund-convexity}.
\end{lemma}

\begin{proof}
Put $a=(\atanh x+\atanh y)/2$ and $\ell=F^{-1}(e/d)$.
Lemma~\ref{lem:edge-entropy-profile} gives $b\ge\ell$.
The identities $x=\tanh(a+b)$ and $y=\tanh(a-b)$ imply
\[
 d=\frac{\sinh b\cosh b}{\cosh^2a+\sinh^2b},\qquad
 A=\arctan\!\left(\frac{\sinh b}{\cosh a}\right).
\]
For $t>0$, the function $q(t)=(1+t^2)(\arctan t/t)^2$
is strictly increasing: its logarithmic derivative is
$2[(\arctan t)^{-1}-t^{-1}]/(1+t^2)>0$. Hence
\[
 \frac{A^2}{d}
 =\tanh b\,q\!\left(\frac{\sinh b}{\cosh a}\right)
 \le\tanh b\,q(\sinh b)
 =\frac{\zeta(b)^2}{\tanh b}.
\]
Hence $db-A^2\ge dR(b)\ge dR(\ell)$ by
Lemma~\ref{fh:refund-convexity}.
The first inequality is strict unless $a=0$, equivalently $c=0$;
in that case $e/d=F(b)$ and both inequalities are equalities.
\end{proof}

\begin{proof}[Proof of Proposition~\ref{fh:simultaneous-refund}]
Write $e=E_\rho(f)$. On coordinate fibers let $e_i$ be the average
endpoint entropy and $q_i=(g_+-g_-)/2$ the signed half-difference.
Then $\E e_i=e$ and $\E q_i=\rho\widehat f(i)$.
Average the inequality in Lemma~\ref{fh:edge-inequality} and apply Jensen's inequality
to the convex even extension of $\mathcal R$ from \eqref{eq:edge-remainder}. For $|\widehat f(i)|>0$ this gives
\[
 \mathcal D_i(g)-\operatorname{Dir}_i(h)
 \ge\rho |\widehat f(i)|R\left(F^{-1}\left(\frac{e}{\rho |\widehat f(i)|}\right)\right)
 \quad(|\widehat f(i)|>0).
\]
The other coordinates have nonnegative remainders by edge
Cauchy--Schwarz, including those with $|\widehat f(i)|=0$. Summing proves the first inequality.
Positive gap gives $e<H(\rho)$; monotonicity in $e$ therefore gives the positive-gap bound.

Finally, Proposition~\ref{prop:profile-production} gives
$D(g)\ge\sum_i B_e(\rho |\widehat f(i)|)$. Positive gap and monotonicity in $e$
bound this sum below by $\sum_i B_{H(\rho)}(\rho |\widehat f(i)|)
=\rho\sum_i |\widehat f(i)|\ell_i$, proving the profile budget.  
\end{proof}

\subsection{Spectral completion}\label{sec:spectral-completion}

Completing a square transfers transformed energy to the noisy
source. It also yields entropy inequalities for the finite reversible
Markov semigroups defined in Section~\ref{sec:markov-semigroups}.
We retain their general form because the argument uses the reversible
edge identity and the spectrum of the generator. We then return to
the cube and the projection argument.

\begin{theorem}[Entropy bounds under reversible noise]\label{fh:entropy-smoothing}
Let $P_s=e^{-s\mathcal L}$ be a Markov semigroup on a finite
probability space $(\Omega,\nu)$, reversible with stationary measure $\nu$.
Expectations use $\nu$, and
$\operatorname{Dir}(h)=\E[h\mathcal Lh]$ for $h:\Omega\to\R$.
For every $f:\Omega\to[-1,1]$ and $s>0$,
\[
\begin{split}
 \E H(P_sf)&\ge(\log2)(1-\E f^2)
   +\frac{9}{2s^{3/2}}\int_0^s t^{3/2}\operatorname{Dir}(P_tf)\,dt,\\
 \E\sqrt{1-(P_sf)^2}&\ge1-\E f^2
   +\frac5{2\sqrt s}\int_0^s\sqrt t\,\operatorname{Dir}(P_tf)\,dt.
\end{split}
\]
\end{theorem}

\begin{lemma}[Spectral transfer]\label{fh:square-completion}
Let $A$ be self-adjoint on a finite-dimensional real Hilbert space.
If $\gamma>0$ and $A+\gamma I$ is positive definite, then
\[
 \langle v,Av\rangle+\gamma\|v-w\|^2
 \ge\langle w,\gamma A(\gamma I+A)^{-1}w\rangle
\]
for all $v,w$.
\end{lemma}
\begin{proof}
In an orthonormal eigenbasis of $A$, use the scalar identity
\[
 \lambda x^2+\gamma(x-y)^2
 =(\lambda+\gamma)\left(x-\frac{\gamma y}{\lambda+\gamma}\right)^2
   +\frac{\gamma\lambda}{\gamma+\lambda}y^2,
\]
and discard the square. Negative eigenvalues are allowed as long
as $\lambda> -\gamma$.
\end{proof}

\begin{theorem}[Entropy and smoothed energy]\label{fh:markov-energy}
Let $P_s=e^{-s\mathcal L}$ be a Markov semigroup on a finite
probability space $(\Omega,\nu)$, reversible with stationary measure $\nu$.
Write $\operatorname{Dir}(h)=\E[h\mathcal Lh]$ for $h:\Omega\to\R$,
and $D(h)=\E[(\mathcal Lh)\atanh h]$ for $h:\Omega\to(-1,1)$.
For every $f:\Omega\to[-1,1]$ and $s>0$,
\[
 \frac92\operatorname{Dir}(P_sf)
 \le D(P_sf)+\frac{3}{2s}
 \left[\E H(P_sf)-(\log2)(1-\E f^2)\right].
\]
Endpoint values are understood by approximation.
\end{theorem}
\begin{proof}
Set $g=P_sf$. Edgewise Cauchy--Schwarz gives
$\operatorname{Dir}(\arcsin g)\le D(g)$ for every reversible
Markov generator. Lemma~\ref{fh:arcsine-error} gives
$\|\arcsin g-\frac\pi2g\|_2^2\le(12/7)\E\Psi(g)$.
Use Lemma~\ref{fh:square-completion} with $\gamma=7/(8s)$.
This choice makes the interpolation penalty
$\gamma(12/7)\E\Psi(g)=3\E\Psi(g)/(2s)$ exactly match the
entropy coefficient in the conclusion. The identity
\[
 \E H(g)-(\log2)(1-\E f^2)
 =\E\Psi(g)+(\log2)(\|f\|_2^2-\|g\|_2^2)
\]
then separates the transform error from the squared norm lost under
noise. If $f_\lambda$ is the projection onto the eigenspace of
$\mathcal L$ with eigenvalue $\lambda\ge0$, square completion gives
\[
 D(g)+\frac3{2s}\bigl[\E H(g)-(\log2)(1-\E f^2)\bigr]
 \ge\sum_\lambda\left[
 \frac{\pi^2}{4}\frac{\gamma\lambda}{\gamma+\lambda}e^{-2s\lambda}
 +\frac{3\log2}{2s}(1-e^{-2s\lambda})\right]\|f_\lambda\|_2^2.
\]
For a positive eigenvalue, put $x=2s\lambda$ and factor out
$\lambda e^{-2s\lambda}$. The remaining scalar multiplier is
\[
 \frac{7\pi^2/4}{7+4x}
       +3(\log2)\frac{e^x-1}{x}\ge\frac92
 \qquad \text{for } x\ge0.
\]
Indeed, use $\pi^2/4>123/50$, $\log2>9/13$, and
$(e^x-1)/x\ge1+x/2+x^2/6$.
After multiplication by $7+4x$, the difference from $9/2$ is
at least $3(300x^3+1425x^2-525x+56)/650>0$;
the quadratic part has negative discriminant $-43575$.
Summing the scalar bound gives the claimed estimate.
The zero eigenspace contributes zero. Apply the argument first to
$(1-\varepsilon)f$, then let $\varepsilon\downarrow0$.
\end{proof}

\begin{lemma}[Hellinger interpolation]\label{fh:hellinger-interpolation}
Let $J(t)=\sqrt{1-t^2}$ and
$\Theta(t)=\int_0^t(1-x^2)^{-3/4}\,dx$.
For $-1\le t\le1$,
\[
 |\Theta(t)-\Theta(1)t|^2
 \le2[\Theta(1)-1]^2[J(t)-J(t)^2]
 \le6[J(t)-J(t)^2].
\]
\end{lemma}
\begin{proof}
The substitution $x=1-y^4$ gives
$\Theta(1)=4\int_0^1(2-y^4)^{-3/4}\,dy$.
Using $2^{3/4}<17/10$ and the tangent bound
$(1-y^4/2)^{-3/4}\ge1+3y^4/8$ gives
$\Theta(1)>43/17>5/2$.
The chord bound for the convex function $(2-x)^{-3/4}$,
together with $2^{3/4}>5/3$, gives $\Theta(1)<68/25<3$.
In particular, $2[\Theta(1)-1]^2<3698/625<6$.

By symmetry take $0\le t\le1$, and put
$\varphi(t)=\sqrt{J(t)-J(t)^2}$.
Direct differentiation gives
$-\varphi''(t)=3\sqrt{1-J(t)}/[4J(t)^{7/2}]$ and
$\Theta''(t)=3t/[2J(t)^{7/2}]$. Hence
$\Theta''(t)/[-\varphi''(t)]=2\sqrt{1+J(t)}$ decreases strictly
from $2\sqrt2$ to $2$.
It follows that the second derivative of
$w(t)=\sqrt2[\Theta(1)-1]\varphi(t)-\Theta(1)t+\Theta(t)$
changes sign once, from positive to negative.
Since $w(0)=w'(0)=w(1)=0$, its derivative first increases and
then decreases; thus $w$ increases and then decreases, and $w\ge0$.
Convexity gives $\Theta(t)\le\Theta(1)t$, proving the bound.
Finally, the ratio $|\Theta(t)-\Theta(1)t|^2/[J(t)-J(t)^2]$ tends to
$2[\Theta(1)-1]^2$ as $t\to0$, proving optimality.
\end{proof}

\begin{proposition}\label{fh:hellinger-energy}
In the setting of Theorem~\ref{fh:markov-energy}, let
$D_J(h)=\E[(\mathcal Lh)h/J(h)]$.
Then
\[
 \frac52\operatorname{Dir}(P_sf)
 \le D_J(P_sf)+\frac1{2s}\left[\E J(P_sf)-1+\E f^2\right].
\]
\end{proposition}
\begin{proof}
Put $g=P_sf$. Edgewise Cauchy--Schwarz gives
$\operatorname{Dir}(\Theta(g))\le D_J(g)$.
Use the same cancellation as in Theorem~\ref{fh:markov-energy},
now with $\gamma=1/(12s)$ and interpolation constant $6$.
Their product is $1/(2s)$, the coefficient of the Hellinger
remainder in the conclusion.
For $x\ge0$, the scalar multiplier is greater than
$(25/4)/(1+6x)+(e^x-1)/x\ge5/2$.
To verify the last inequality, use
$(e^x-1)/x\ge1+x/2+x^2/6$ and multiply by $12(1+6x)$.
The resulting polynomial is $12x^3+38x^2-102x+57$.
Since $x^3\ge3x-2$, it is at least
$38x^2-66x+33>0$, whose discriminant is $-660$.
The spectral summation from Theorem~\ref{fh:markov-energy}
proves the proposition.
\end{proof}

For Boolean sources on the cube, $P_s=T_{e^{-s}}$ and the endpoint
correction vanishes. Thus
$\frac92\operatorname{Dir}(T_\rho f)
\le D(T_\rho f)+3\E H(T_\rho f)/[-\log(\rho^2)]$.
More generally, differentiating $\E H(P_sf)$ gives $D(P_sf)$.

\begin{proof}[Proof of Theorem~\ref{fh:entropy-smoothing}]
Multiply the inequality in Theorem~\ref{fh:markov-energy} by $s^{3/2}$ and
integrate from zero. Since $\frac d{ds}\E H(P_sf)=D(P_sf)$,
the first inequality follows by an integrating factor.
For the second, use $\frac d{ds}\E J(P_sf)=D_J(P_sf)$ and
multiply the inequality in Proposition~\ref{fh:hellinger-energy} by $\sqrt s$.
The boundary terms at zero vanish because both entropies are bounded.
\end{proof}

\subsection{From variance to entropy production}

The next lemma removes the unknown variance $\Var(g)$ and entropy
error $\mathcal J$ from a spectral bound. At positive gap, it
replaces their combination by an explicit baseline and a nonnegative
correction for the source mean. Here $\mathcal J$ is the entropy
correction defined in \eqref{eq:entropy-correction}.

\begin{lemma}[Transfer of a spectral baseline]\label{ub:transfer}
Let $f:\{-1,1\}^n\to\{-1,1\}$, $0<\rho<1$, and $g=T_\rho f$.
For $\Lambda,\Gamma>0$, define
\begin{equation}\label{eq:mean-transfer}
 \eta_{\Lambda,\Gamma}(\rho)
 =\Gamma\bigl[\tfrac12-(1-{\rho^2})\log2\bigr]-\Lambda {\rho^2}.
\end{equation}
If $G_f(\rho)>0$ and $\eta_{\Lambda,\Gamma}(\rho)\ge0$, then
\[
 \Lambda {\Var(g)}-\Gamma\mathcal J\ge \Lambda {\rho^2}-\Gamma\Psi(\rho)
                    +\eta_{\Lambda,\Gamma}(\rho)\mu^2.
\]
\end{lemma}
\begin{proof}
Parseval gives $\Var(g)\le\rho^2(1-\mu^2)$, and positive gap gives
$0\le\mathcal J<\Var(g)\log2+\Psi(\mu)-\Phi(\rho)$.
The latter inequality uses $\E g^2={\Var(g)}+\mu^2$ and the actual
information inequality $\E\Phi(g)>\Phi(\rho)+\Phi(\mu)$.
The assumption $\eta_{\Lambda,\Gamma}(\rho)\ge0$ implies
$\Lambda-\Gamma {\log2}\le\Gamma(1/2-{\log2})/{\rho^2}<0$.
Since $\Psi(\mu)\le({\log2}-1/2)\mu^2$,
\begin{align*}
 \Lambda {\Var(g)}-\Gamma\mathcal J
 &\ge(\Lambda-\Gamma {\log2}){\Var(g)}+\Gamma\Phi(\rho)-\Gamma\Psi(\mu)\\
 &\ge(\Lambda-\Gamma {\log2}){\rho^2}(1-\mu^2)+\Gamma\Phi(\rho)
       -\Gamma({\log2}-1/2)\mu^2\\
 &=\Lambda {\rho^2}-\Gamma\Psi(\rho)
      +\eta_{\Lambda,\Gamma}(\rho)\mu^2.\qedhere
\end{align*}
\end{proof}

Only $\Var(g)$ is centered; $\mathcal J$ remains the actual
interpolation error. The explicit compact parameters below satisfy the
transfer condition throughout their range by an analytic estimate.

\section{Coordinate projections and the spectral bound}

We prove the uniform spectral criterion that handles every source
below the local cutoff. Retain the coordinates above a threshold and
charge the discarded coordinates to the same profile budget. We first bound the variance of the retained
source and the energy of the discarded arcsine field. Combining these
bounds gives the scalar criterion in Theorem~\ref{as:gate}.

Throughout this section, write $\rho=\tanh u$, $g=T_\rho f$,
and $h=\arcsin g$. We retain the entropy correction $\mathcal J$
from \eqref{eq:entropy-correction} and the profile height $\ell_u$
from \eqref{eq:profile-height}. For $K\subseteq[n]$, define
\begin{equation}\label{eq:coordinate-projection}
 P_Kh=\E[h\mid X_K].
\end{equation}
This averages over the coordinates outside $K$ and is the orthogonal
projection onto functions of $X_K$. In Fourier language it keeps
exactly the terms with $S\subseteq K$; the remainder $h-P_Kh$
contains every term using an outside coordinate. This decomposition
allows separate estimates for dependence carried by the retained
coordinates and for the energy needed to support the remainder.
Throughout this section,
$f:\{-1,1\}^n\to\{-1,1\}$ is nonconstant and increasing.
We retain the full budget $D(g)/\rho$ from
Proposition~\ref{fh:simultaneous-refund}. This gives a direct lower bound for entropy
production at every positive gap.
\subsection{An influence bound for bounded functions}

Conditional sources take values in $[-1,1]$. The following theorem
bounds their variance by an explicit cost for each influence, without
assuming monotonicity.

We express the variance bound in terms of the function
\begin{equation}\label{eq:influence-cost}
 \psi(a)=\min\left\{\frac{a+a^2}{2},
             \frac{2a(1-a/4)}{\log(4/a)}\right\}\quad(0<a\le1),
 \qquad \psi(0)=0.
\end{equation}

\begin{theorem}[Variance and influences]\label{rt:source-projection}
For every $g:\{-1,1\}^n\to[-1,1]$ and every $K\subseteq[n]$,
\[
 \Var(P_Kg)\le\sum_{i\in K}\psi(\|\partial_i g\|_1).
\]
In particular, $\Var(g)\le\sum_i\psi(\|\partial_i g\|_1)$.
\end{theorem}

The estimate is a quantitative form of the following influence inequality.

\begin{theorem}[Talagrand's influence inequality
{\cite{Talagrand1994,CorderoEskenazis2023}}]
There is a universal constant $C>0$ such that, for every
$f:\{-1,1\}^n\to\R$,
\[
 \Var(f)\le C\sum_{i=1}^n
 \frac{\|\partial_i f\|_2^2}
      {1+\log(\|\partial_i f\|_2/\|\partial_i f\|_1)}.
\]
\end{theorem}

We now prove the explicit bounded-function estimate used here.
\begin{proof}[Proof of Theorem~\ref{rt:source-projection}]
Fourier expansion gives
$\Var(g)=\sum_i\int_0^1\|T_{\sqrt s}\partial_i g\|_2^2\,ds$.
Indeed, a mode of degree $j$ contributes $s^{j-1}$ in each of its
$j$ coordinates. Fix $i$ and write $a=\|\partial_i g\|_1$.
Since $|\partial_i g|\le1$, the integrand is at most
$(1-s)a^2+sa$ by its Fourier expansion, and at most
$a^{2/(1+s)}$ by hypercontractivity.
Integrating the first bound gives $(a+a^2)/2$.
For the second, substitute $r=(1-s)/(1+s)$ and use
$\log(1+r)\ge r\log2$ on $[0,1]$:
\[
 \int_0^1a^{2/(1+s)}\,ds
 =2a\int_0^1\frac{a^r}{(1+r)^2}\,dr
 \le2a\int_0^1(a/4)^r\,dr
 =\frac{2a(1-a/4)}{\log(4/a)}.
\]
Both bounds apply to each coordinate's integral separately, so
we may take their minimum before summing. Zero influences contribute
zero. For a projection, apply the bound to $P_Kg$: conditional
expectation contracts each derivative in $L^1$, and $\psi$ is increasing.
\end{proof}

For monotone $g$, $\partial_i g\ge0$, so
$\|\partial_i g\|_1=\widehat g(i)$. This is the form used for the Boolean source below.

\subsection{Retaining the pivot energy in a martingale decomposition}

We next bound the energy of the part discarded by projection,
using the martingale entropy argument of Falik and Samorodnitsky
\cite[Theorem~2.2 and Lemma~2.3]{FalikSamorodnitsky2007}.
Reveal the coordinates one at a time and subtract successive
conditional expectations. These differences are the Doob martingale
increments: each records the additional dependence revealed at that
step. Its Fourier terms all contain the newly revealed coordinate,
so it already contributes one unit of degree before the other
coordinates are counted. Keeping that contribution strengthens the
log-Sobolev estimate used below.

\begin{lemma}[Martingale energy]
\label{as:strengthened-split}
Let $h:\{-1,1\}^n\to\R$ and $K\subseteq[n]$.
Set 
\[
V=\|h-P_Kh\|_2^2 \enskip \text{and }
b_K=\sum_{i\notin K}(\E|\partial_i h|)^2.
\]
For every real $\lambda$, we have
\[
 \operatorname{Dir}(h-P_Kh)
 \ge\left(1+\frac\lambda2\right)V
       -\frac{e^{\lambda-1}}2 b_K.
\]
When $V$ and $b_K$ are positive, the stronger nonlinear form is
\[
 \operatorname{Dir}(h-P_Kh)
 \ge V+\frac V2\log\frac V{b_K}.
\]
\end{lemma}

\begin{proof}
Order the coordinates of $K$ before its complement, and let
$h_i:\{-1,1\}^n\to\R$ be the Doob increments of $h$ in this
order, so
$h_i=\E[h\mid X_1,\ldots,X_i]-\E[h\mid X_1,\ldots,X_{i-1}]$.
Each has the form $h_i=X_i g_i(X_{<i})$, with
$g_i:\{-1,1\}^{i-1}\to\R$. Thus
$\Ent(h_i^2)=\Ent(g_i^2)$ and
$\operatorname{Dir}(h_i)=\|h_i\|_2^2+\operatorname{Dir}(g_i)$.
The log-Sobolev part of Theorem~\ref{ext:cube-functional}, applied
to $g_i$, gives
$\Ent(h_i^2)\le2[\operatorname{Dir}(h_i)-\|h_i\|_2^2]$.
Also $h_i=\E[\mathcal L_i h\mid X_1,\ldots,X_i]$, so
$\E|h_i|\le\E|\partial_i h|$.

For a real random variable $Z$ with $v=\E Z^2>0$, use the
probability measure $dQ=Z^2\,dP/v$. Jensen gives
$\E_Q\log|Z|\ge-\log\E_Q|Z|^{-1}=\log(v/\E|Z|)$, hence
$\Ent(Z^2)\ge v\log[v/(\E|Z|)^2]$.
Zeros of $Z$ have zero $Q$-measure, so approximation justifies this
calculation.
Applying this bound and the log-sum inequality to the increments
whose pivot is outside $K$ yields
$\sum_{i\notin K}\Ent(h_i^2)\ge V\log(V/b_K)$.
Their Fourier supports are disjoint: a nonconstant mode belongs
to the increment at its last revealed coordinate. Hence their squared
norms sum to $V$ and their energies sum to
$\operatorname{Dir}(h-P_Kh)$.
This proves the second inequality. The inequality
$V\log(V/b_K)\ge\lambda V-e^{\lambda-1}b_K$ gives the first inequality. If $b_K=0$, every outside derivative
vanishes, so $V=0$; when $V=0$ the latter inequality is immediate.
\end{proof}

\subsection{A signed Fourier operator and the source projection}

We combine the discarded-field estimate with the retained-source
variance bound. The resulting inequality assigns one cost to each
retained coordinate; the discarded singletons are paid for by the
unused part of the common profile budget.

There are two choices in the comparison. The parameter $t$ is the
energy level assigned to discarded modes of degree at least two;
$\gamma$ is the penalty for replacing the arcsine transform by a
multiple of the noisy source. It must also offset any negative
singleton multiplier. Accordingly, for $t\ge2$, choose
\begin{equation}\label{as:dual-parameters}
 \gamma>\max\{0,\tfrac12e^{2t-3}-t\},\qquad
 k(s)=\frac{\gamma s}{\gamma+s}.
\end{equation}
For a selected coordinate set $K$, Lemma~\ref{as:strengthened-split},
applied to $h$ with $\lambda=2t-2$, gives
\[
 \operatorname{Dir}(h)
 \ge \operatorname{Dir}(P_Kh)
   +t\|h-P_Kh\|_2^2
   -\tfrac12e^{2t-3}\sum_{i\notin K}\widehat h(i)^2.
\]
Here monotonicity of $h$ implies
$\E|\partial_i h|=\widehat h(i)\ge0$.
The quadratic form on the right defines a self-adjoint operator
$A\chi_S=\lambda_S\chi_S$. Its Fourier eigenvalues are
\begin{equation}\label{as:operator-multipliers}
 \lambda_S=\begin{cases}
 0,&S=\varnothing,\\
 |S|,&\varnothing\ne S\subseteq K,\\
 t-\tfrac12e^{2t-3},&S=\{i\},\ i\notin K,\\
 t,&|S|\ge2,\ S\not\subseteq K.
 \end{cases}
\end{equation}
The multiplier records how much energy the preceding estimate
assigns to each Fourier mode. A mode entirely inside $K$ retains
its actual degree, while an outside mode of degree at least two is
assigned the common lower bound $t$. Outside singletons incur a
correction and may have a negative multiplier. The choice of
$\gamma$ ensures that adding $\gamma I$ makes the operator positive,
so the square-completion lemma applies with the transferred
multiplier $k$ from \eqref{as:dual-parameters}.

\begin{lemma}\label{as:spectral-line}

With $R$ and $C_*$ from \eqref{eq:arcsine-remainder} and
\eqref{eq:arcsine-constant}, if $G_f(\rho)>0$,
\[
 {D(g)}\ge \rho\sum_{i\in K}\widehat f(i) R(\ell_i)-\gamma C_*\mathcal J
       +\frac{\pi^2}{4}\sum_S k(\lambda_S)\rho^{2|S|}\widehat f(S)^2.
\]
\end{lemma}

\begin{proof}
The simultaneous remainder bound in Proposition~\ref{fh:simultaneous-refund}
first gives
${D(g)}\ge\rho\sum_{i\in K}\widehat f(i)R(\ell_i)+\operatorname{Dir}(h)$.
Every multiplier
in \eqref{as:operator-multipliers} is greater than $-\gamma$, even
when an outside singleton multiplier is negative. Apply
Lemma~\ref{fh:square-completion} to this operator, with
$v=h$ and $w=\frac{\pi}{2}g$,
then use $\|h-\frac{\pi}{2}g\|_2^2\le C_*\mathcal J$ from
Lemma~\ref{fh:arcsine-error}.
This proves the lemma. The constant mode of
$h$ need not vanish. It is included in the error norm,
whereas its quadratic-form multiplier is zero. Here $k(0)=0$ even when $\widehat g(\varnothing)=\mu\ne0$.
\end{proof}

Theorem~\ref{rt:source-projection}, applied to $f$ and $K$, gives $\Var(P_Kf)\le\sum_{i\in K}\psi(\widehat f(i))$.

Choose $0<z<u$ and retain the coordinates above the threshold
\begin{equation}\label{eq:retained-threshold}
 \beta=\frac{F(u)}{F(z)},\qquad
 K=\{i:\ell_i\ge z\}=\{i:\widehat f(i)\ge\beta\}.
\end{equation}
Define the nonnegative coefficients
\begin{equation}\label{eq:spectral-corrections}
 \begin{split}
 \Delta_2&=[k(t)-k(2)]\rho^4,\\
 \Delta_1&=[k(t)-k(1)]{\rho^2},\\
 \Delta_o&=[k(t)-k(t-\tfrac12e^{2t-3})]{\rho^2},\qquad
 \tau=\Delta_o\frac\beta z.
 \end{split}
\end{equation}
The inequalities $t\ge2$ and monotonicity of $k$ imply
$\Delta_1\ge\Delta_2\ge0$.

\begin{proposition}\label{as:master-ledger-equation}

If $G_f(\rho)>0$,
\[
 \begin{split}
 {D(g)}\ge{}&\frac{\pi^2}{4}k(t){\Var(g)}-\gamma C_*\mathcal J-\frac{\pi^2}{4}\tau\frac{D(g)}{\rho}\\
 &+\sum_{i\in K}\widehat f(i)\left[
    \rho R(\ell_i)+\frac{\pi^2}{4}\tau\ell_i
    -\frac{\pi^2}{4}\Delta_2\frac{\psi(\widehat f(i))}{\widehat f(i)}
    -\frac{\pi^2}{4}(\Delta_1-\Delta_2)\widehat f(i)\right].
 \end{split}
\]
\end{proposition}

\begin{proof}
Relative to the constant multiplier $k(t)$ on nonconstant modes,
the raw Fourier accounting is
\begin{align*}
{D(g)}\ge{}&\frac{\pi^2}{4}k(t){\Var(g)}-\gamma C_*\mathcal J
 +\rho\sum_{i\in K}\widehat f(i)R(\ell_i)\\
&-\frac{\pi^2}{4}\left[
 \Delta_2\Var(P_Kf)
 +(\Delta_1-\Delta_2)\sum_{i\in K}\widehat f(i)^2
 +\Delta_o\sum_{i\notin K}\widehat f(i)^2\right].
\end{align*}
Indeed, the singleton coefficient inside $K$ is $\Delta_1$. For
every degree $j\ge2$,
$[k(t)-k(j)]\rho^{2j}\le\Delta_2$: a nonpositive coefficient may be dropped,
and otherwise $k(j)\ge k(2)$ and $\rho^{2j}\le \rho^4$. Thus the cost is
at most
\[
 \Delta_2\Var(P_Kf)+(\Delta_1-\Delta_2)\sum_{i\in K}\widehat f(i)^2
 \le\sum_{i\in K}\bigl[\Delta_2\psi(\widehat f(i))
                    +(\Delta_1-\Delta_2)\widehat f(i)^2\bigr].
\]
The remaining cost is exactly
$\Delta_o\sum_{i\notin K}\widehat f(i)^2$.

To bound this cost, Lemma~\ref{as:profile-convex} shows that
$\ell_u(a)/a$ is decreasing: at $v=\ell_u(a)$ it equals
$vF(v)/F(u)$, while $v$ increases with $a$.
Consequently $\widehat f(i)^2\le(\beta/z)\widehat f(i)\ell_i$ outside $K$.
The profile budget in Proposition~\ref{fh:simultaneous-refund} gives
$\sum_{i\notin K}\widehat f(i)\ell_i\le D(g)/\rho-\sum_{i\in K}\widehat f(i)\ell_i$.
Insert these two bounds into the inequality in Lemma~\ref{as:spectral-line} and collect
terms. This proves the proposition.
\end{proof}

\subsection{One bound for all retained coordinates}

We now compare the spectral baseline with the cost of each unit of
the profile budget. This treats every retained coordinate uniformly,
so the largest influence enters only as an upper bound on the range
of a scalar function.

The preceding bound has one spectral baseline and one correction
for each retained coordinate. Dividing that coordinate's correction
by its profile contribution $a\ell_u(a)$ gives the scalar cost
below. Bounding every such cost by the same number lets us sum
over an arbitrary number of coordinates using the common production
budget. Using $\psi$ from \eqref{eq:influence-cost} and the corrections
in \eqref{eq:spectral-corrections}, define, for $0<a\le1$,
\begin{equation}\label{eq:spectral-cost}
 \mathcal Q_u(a)=
 \frac{\rho[\ell_u(a)-R(\ell_u(a))]
       +\frac{\pi^2}{4}\Delta_2\psi(a)/a
       +\frac{\pi^2}{4}(\Delta_1-\Delta_2)a}{\ell_u(a)}.
\end{equation}

\begin{theorem}[Uniform spectral bound]\label{as:gate}
Let $f:\{-1,1\}^n\to\{-1,1\}$ be nonconstant and increasing,
$\rho=\tanh u$ with $u>0$, and $\alpha=\max_i|\widehat f(i)|$.
Choose $0<z<u$, $t\ge2$, and $\gamma$ satisfying
\eqref{as:dual-parameters}. With $\beta,\tau,\mathcal Q_u$ as in
\eqref{eq:retained-threshold}--\eqref{eq:spectral-cost}, put
\[
 b=\max\left\{\rho+\frac{\pi^2}{4}\tau,
              \sup_{\beta\le a\le\alpha}\mathcal Q_u(a)\right\}.
\]
If $G_f(\rho)>0$ and the mean-transfer coefficient from
\eqref{eq:mean-transfer} satisfies
$\eta_{\frac{\pi^2}{4}k(t),\gamma C_*}(\rho)\ge0$, then
\[
 \frac{D(T_\rho f)}{\rho}
 \ge\frac{\frac{\pi^2}{4}k(t)\rho^2-\gamma C_*\Psi(\rho)}{b}.
\]
At such a positive gap, $G_f'(\rho)>0$ whenever
\[
 \frac{\frac{\pi^2}{4}k(t)\rho^2-\gamma C_*\Psi(\rho)}{u}
 >\max\left\{\rho+\frac{\pi^2}{4}\tau,
              \sup_{\beta\le a\le\alpha}\mathcal Q_u(a)\right\}.
\]
\end{theorem}
\begin{proof}
Apply Lemma~\ref{ub:transfer} to
Proposition~\ref{as:master-ledger-equation}. Its retained contribution is
$\sum_{i\in K}\widehat f(i)\ell_i[\rho+\pi^2\tau/4-\mathcal Q_u(\widehat f(i))]$.
Because $\rho+\pi^2\tau/4-b\le0$, the profile budget in Proposition~\ref{fh:simultaneous-refund} bounds this below by
$(\rho+\pi^2\tau/4-b)D(g)/\rho$.
Consequently
\[
 D(g)\ge\frac{\pi^2}{4}k(t)\rho^2-\gamma C_*\Psi(\rho)
                  +(\rho-b)\frac{D(g)}{\rho}.
\]
Rearrange and use $G_f'(\rho)=D(g)/\rho-u$.
\end{proof}

A condition proved for all $a\le\alpha$ also covers every smaller
maximum influence. This simple monotonicity eliminates the partition
in the source parameter. Lemma~\ref{ub:exponential-local} covers
$\alpha\ge1-2e^{-u}$ throughout the compact range. The spectral
calculation therefore uses this moving cutoff throughout as well.
No separate uniform dictator-neighborhood theorem is needed for
the qualitative argument.

\subsection{An analytic bound for the spectral cost}

Supporting lines to the convex remainder $R$ give convex upper bounds
for the spectral cost on every positive height interval. This reduces
the remaining calculation to endpoint values. The key is a global
shape property of the entropy profile.

\begin{lemma}[Convex upper bounds for the scalar costs]
\label{as:convex-majorant}
Fix $u,c>0$ and $\Delta_1\ge\Delta_2\ge0$, and put $\rho=\tanh u$.
In the cost \eqref{eq:spectral-cost}, replace $\psi(a)$ by either
$(a+a^2)/2$ or $2a(1-a/4)/\log(4/a)$.
Set $a=F(u)/F(v)$ and replace $[v-R(v)]/v$ by
\[
 1-R'(c)+\frac{cR'(c)-R(c)}v.
\]
The resulting function is a convex upper bound for
$\mathcal Q_u(F(u)/F(v))$ on $0<v\le u$.
The latter replacement increases the expression by at most
$\rho(v-c)^2/(4v)$.
With the same choice for $\psi$, replacing
$\rho[v-R(v)]/v$ instead by $\pi^2/(4v)$ also gives
a convex upper bound.
\end{lemma}

\begin{proof}
Put $a=F(u)/F(v)$. Since $F$ is decreasing, $0<v\le u$ implies
$0<a\le1$ and $\ell_u(a)=v$. Either formula in
\eqref{eq:influence-cost} is at least $\psi(a)$; its coefficient
in \eqref{eq:spectral-cost} is nonnegative.
Convexity of $R$ gives $R(v)\ge R(c)+R'(c)(v-c)$ and
$cR'(c)-R(c)\ge0$, since $R(0)=0$. The tangent replacement is
therefore an upper bound consisting of a constant and a nonnegative
multiple of $1/v$.

For the quadratic choice, the remaining terms are
\[
 \frac{\pi^2}{4}\left[\frac{\Delta_2}{2v}
 +\left(\Delta_1-\frac{\Delta_2}{2}\right)
   \frac{F(u)}{vF(v)}\right].
\]
For the logarithmic choice, use
$2(1-a/4)/\log(4/a)=2\int_0^1(a/4)^r\,dr$ to obtain
\[
 \frac{\pi^2}{4}\left[
 2\Delta_2\int_0^1\frac{(F(u)/4)^r}{vF(v)^r}\,dr
 +(\Delta_1-\Delta_2)\frac{F(u)}{vF(v)}\right].
\]
Both are nonnegative combinations of the convex functions
$1/[vF(v)^r]$, $0\le r\le1$, from Lemma~\ref{as:profile-convex}.
Integration preserves their convexity, proving the first assertion.

For the error estimate, put $p=\arctan(\sinh v)/\sinh v$.
The bounds $x/(1+x^2)\le\arctan x\le x$ give
$\operatorname{sech}^2v\le p\le1$. Differentiating $R'=(1-p)^2$ yields
\[
 0\le R''(v)=2\coth v(1-p)(p-\operatorname{sech}^2v)
 \le\tfrac12\tanh^3v\le\tfrac12,
\]
because the two nonnegative factors have sum $\tanh^2v$.
Taylor's theorem now gives
$0\le R(v)-R(c)-R'(c)(v-c)\le(v-c)^2/4$.
Multiplication by $\rho/v$ proves the stated increase from the
tangent replacement; it does not include the earlier choice for $\psi$.

Finally, convexity of $\arcsin$ on $[0,1]$ gives
$\arcsin x\le(\pi/2)x$. By \eqref{eq:arcsine-remainder},
\[
 \rho[v-R(v)]
 =\rho\frac{\arcsin^2(\tanh v)}{\tanh v}
 \le\frac{\pi^2}{4}\rho\tanh v\le\frac{\pi^2}{4}.
\]
Replacing the first term by $\pi^2/(4v)$ thus leaves a sum of
convex functions, proving the alternative bound.
\end{proof}

A fixed choice of the formula for $\psi$ and of $c$ therefore requires
only two endpoint checks
on any closed height interval. Choosing $c=\sqrt{v_-v_+}$ on
$[v_-,v_+]$ bounds the tangent error uniformly by
$\rho(\sqrt{v_+}-\sqrt{v_-})^2/4$.
This follows by maximizing $(v-c)^2/v$ at the endpoints.
More generally, intervals equally spaced in $\sqrt v$ make this
error decay quadratically in their number.

\subsection{An explicit choice on the compact noise range}\label{sec:compact-spectral}

The inverse profile determines the upper end of the retained height
range. Bounds on the growth of $e^{2v}F(v)$ give an explicit upper
height and control its influence. The compact calculation therefore
needs no inverse roots or curvature estimates.

\begin{lemma}[Inverse bound]\label{as:explicit-inverse}
Let $u>2$ and $e^{2-u}\le a\le1$. For
$v=u+(2u+1)\log a/(4u+1)$,
\[
 F^{-1}(F(u)/a)\le v\le u,\qquad
 \frac{F(u)}{F(v)}\le e^{-2(u-v)}\frac{2u+1}{2v+1}.
\]
\end{lemma}
\begin{proof}
Define the scaled profile
\begin{equation}\label{eq:scaled-profile}
 S(s)=e^{2s}F(s).
\end{equation}
Lemma~\ref{app:scaled-profile-derivative} gives
$1/(2s+1)<S'(s)/S(s)<2/(2s+1)$ for $s\ge2$.
The case $a=1$ is exact. Otherwise put $w=F^{-1}(F(u)/a)<u$.
Since $-F'/F>1$,
$-\log a=\int_w^u(-F'/F)\,ds>u-w$,
so $w>u+\log a\ge2$. On $[w,u]$ we also have
$-F'/F=2-S'/S<2-1/(2u+1)$. Integrating gives
$-\log a<(2-1/(2u+1))(u-w)$, hence $w<v\le u$.
Finally, integration on $[v,u]$ gives
$\log[S(u)/S(v)]\le\log[(2u+1)/(2v+1)]$.
Using \eqref{eq:scaled-profile} proves the second bound.
\end{proof}

Throughout $45951/20000\le u\le8$, use the single choice
\[
 z=\log u+\frac{u}{25}-\frac25,\qquad
 t=\frac{49u}{100}+\frac75,\qquad
 \gamma=\frac12e^{49u/50-1/5}+e^{9u/20+1/4}.
\]
Thus $\gamma>e^{2t-3}/2$ identically. Also $0<z<2u/3$ and
$t>2$, using $\log2>1/2$ and $\log u\le u/2$ for $u\ge2$.
The transfer condition is analytic as well. The ratio $\gamma/t$ increases because
$\gamma'/\gamma\ge9/20>t'/t$, and the exponential series gives
$\gamma(9/4)>85/12>(5/2)t(9/4)$. Since
$1-\rho^2<1/20$, $C_*>8/5$, $\pi^2/4<5/2$, and
$\log2<7/10$, the normalized transfer margin exceeds
$(8/5)(1/2-7/200)-5/7=26/875$.

Set $c(u)=1-2e^{-u}$ throughout this interval. This cutoff satisfies $u+\log c(u)>2$:
$u>9/4$, and $e^2>22/3$, $e^{1/4}>41/32$ give
$c(u)>e^{-1/4}$. Let
$v_+(u)=u+(2u+1)\log c(u)/(4u+1)$.
By Lemma~\ref{as:explicit-inverse}, this explicit height covers
all $a\le c(u)$, and $v_+(u)>u-1/4>2u/3$.

On $[z,2u/3]$, use the logarithmic formula for $\psi$ in
Lemma~\ref{as:convex-majorant}, with the tangent to $R$ at $z$.
On $[2u/3,v_+(u)]$, use the quadratic formula with the tangent
at $v_+(u)$. Each bound is convex in height, so its two endpoints
suffice. At the upper endpoint, substitute the upper bound for
$F(u)/F(v_+(u))$ from Lemma~\ref{as:explicit-inverse}; both
influence costs increase with this argument.

The certificate checks these four height endpoints and the
discarded-coordinate inequality on 24 closed noise intervals.
All parameters are explicit functions of $u$; there are no fitted
parameter families or stored height records. Fourth-order Taylor
bounds retain cancellation on each noise interval. The separate
audit checks the original costs throughout the extended height
range. Appendix~\ref{app:certificate} records the replay scope.

\section{Completion of the qualitative proof}\label{sec:qualitative-completion}

The same uniform spectral bound can reach the local entropy region
even as that region shrinks toward a dictator. We first state the
local threshold, then choose spectral parameters whose losses remain
small enough to meet it. Together these estimates complete the proof of CK.

\subsection{A local junction valid for both error masses}

The local entropy comparison covers sufficiently small dictator
distance at each correlation. The following explicit boundary will
also be the starting point of the spectral comparison.

Recall the distance $\delta(f)$ to a signed dictator from
\eqref{eq:dictator-distance}. To describe the local region, put
\begin{equation}\label{eq:local-thresholds}
 u_\delta=\log\frac{3}{2\delta}\quad(\delta>0),\qquad
 d_-(u)=\frac32e^{-u}.
\end{equation}
Thus $u\le u_\delta$ is equivalent to $\delta\le d_-(u)$.

\begin{lemma}\label{rt:local-junction}
Let $f:\{-1,1\}^n\to\{-1,1\}$ satisfy $0<\delta(f)\le1/64$.
Then CK holds for $0\le\rho\le\tanh u_{\delta(f)}$,
at arbitrary output mean.
\end{lemma}
\begin{proof}
Put $\delta=\delta(f)$. By Lemma~\ref{ub:local-propagation},
it suffices to prove the local comparison at $\rho=\tanh u_\delta$.
Let $p=4\delta^2/(9+4\delta^2)$ and
$\rho_\delta=1-2p=\tanh u_\delta$.
Elementary entropy bounds give
$H(\rho_\delta(1-2\delta))\ge H_{\rm b}(\delta)$,
$H_{\rm b}(\delta)/\delta\ge \log(1/\delta)+1-\delta$, and
$H_{\rm b}(p)/p\le2\log(1/\delta)+\log(9/4)+4\delta^2/9+1$.
Using the local deficit \eqref{eq:local-deficit} and dividing
$M_{\rho_\delta}(1-2\delta)$ by $2\delta p$,
its positivity follows from
\[
 2(1-4\delta^2/9)(\log(1/\delta)+1-\delta)
 >(1+8\delta/9)(2\log(1/\delta)+\log(9/4)+4\delta^2/9+1).
\]
The difference is at least
\[
 1-\log(9/4)
 -\delta[2+(8/9)(2\log(1/\delta)+\log(9/4)+1)]
 -(4\delta^2/9)(2\log(1/\delta)+3).
\]
Use $\log64<25/6$ and $\log(9/4)<13/16$.
The functions $\delta \log(1/\delta)$ and $\delta^2\log(1/\delta)$ increase in this
range, so substitution at $\delta=1/64$ bounds the difference
below by $65/4608>1/80$.
Theorem~\ref{ub:local} now gives the claim.
\end{proof}

\subsection{A spectral comparison reaching the local neighborhood}

For the qualitative argument, choose the square-completion parameter
on the scale $e^u$, with spectral level $t$ on the scale $u/2$.
Its entropy penalty is then of order $ue^{-u}$, the same order as
the spectral gain at the local distance threshold. Two simple savings
make these estimates overlap: we keep the arcsine correction and
combine the degree corrections before estimating their remaining loss.
The stronger uniform margins used for stability are established in
 the companion paper \cite{VuTranStability}.

\begingroup
\raggedbottom

\begin{proposition}\label{qual:tail-derivative}
Let $f:\{-1,1\}^n\to\{-1,1\}$ be increasing, and let
$\rho=\tanh u$ with $u\ge8$. If
$\delta(f)\ge(3/2)e^{-u}$, then
$G_f'(\rho)>0$ whenever $G_f(\rho)>0$.
\end{proposition}

\begin{proof}
Apply Theorem~\ref{as:gate} with
$z=7/2, t=u/2+1$, and $\gamma=e^u/4$.
Write $h=e^{-u}$, $a_0=1-3h$, and
$N=(\pi^2/4)k(t)\rho^2-\gamma C_*\Psi(\rho)$ for the
spectral baseline. The distance assumption gives $\alpha\le a_0$.
We will compare the baseline with the discarded-coordinate cost
and every retained-coordinate cost on $\beta\le a\le a_0$.

First, the arcsine contribution has a useful saving. The function
$\arcsin^2 x/x$ increases on $(0,1)$, so for $0<v\le u$ we have
\[
 \rho[v-R(v)]\le\arcsin^2\rho
 =\frac{\pi^2}{4}-2\arctan h\,\bigl(\pi-2\arctan h\bigr)
 <\frac{\pi^2}{4}-4h.
\]
Here $h<1/2500$, since $e>27/10$, and
$h/(1+h^2)\le\arctan h\le h$ proves the last inequality.
Replacing the arcsine numerator by this positive constant gives a
convex upper cost for either formula for $\psi$, as in
Lemma~\ref{as:convex-majorant}. Use the logarithmic formula through
$a=1/8$ and the quadratic formula thereafter. Only the heights
$7/2$, $\ell_u(1/8)$, and $\ell_u(a_0)$ need to be checked.

We use $12/5<\pi^2/4<5/2$ and $1<C_*<12/7$.
The parameter condition follows from $e>2$, and the transfer
condition follows from $\gamma>10t$ and $1-\rho^2<4h^2$.
The bounds
\[
 k(t)\rho^2>t-\frac1{10},\qquad
 \gamma C_*\Psi(\rho)<\frac1{300}
\]
give $N/u>6/5+2/u$. Indeed,
$t^2/\gamma=(u+2)^2e^{-u}$ and $4te^{-2u}$ decrease from $u=8$,
and $\Psi(\rho)\le H(\rho)\le(2u+1)e^{-2u}$.

For discarded coordinates,
\[
 k(t)-k(t-e^{2t-3}/2)
 =\frac{\gamma^2 e^{2t-3}/2}
        {(\gamma+t)(\gamma+t-e^{2t-3}/2)}
 <\frac{e^u}{2(e-2)}.
\]
Together with \eqref{rt:scaled-profile}, this gives
\[
 \beta\le\frac{1001}{1000}\frac{2u+1}{8}e^{7-2u},\qquad
 \tau<\frac{1001}{1000}\frac{(2u+1)e^{7-u}}{56(e-2)}.
\]
The upper bound for $u\tau$ decreases for $u\ge8$.
Using $e>27/10$ at $u=8$ therefore gives
$N-u(\rho+\pi^2\tau/4)>1/3$ throughout this range.

For the first retained height, the profile bound gives
$\beta<1/200$ and $\log(4/\beta)\ge(3u-5)/2$.
For the latter inequality, compare
$2u-7+\log(32000/[1001(2u+1)])$ with $(3u-5)/2$:
the difference is positive at $8$ and has positive derivative.
Since $\Delta_2<t-1$ and $\Delta_1-\Delta_2<2$, the cost at
$v=7/2$ is less than
\[
 \frac57\left(1+\frac{2u}{3u-5}+\frac1{100}\right)
 <\frac65+\frac2u.
\]
At $a=1/8$, the profile derivative first gives
$\ell_u(1/8)>u-\log8>4$, then $\ell_u(1/8)>u-6/5$. The quadratic cost is less than
\[
 \frac52\frac{5/4+9u/32}{u-6/5}<\frac65+\frac2u.
\]
For both comparisons, clearing positive denominators gives a
quadratic that is positive and increasing for $u\ge8$.
The logarithmic cost is smaller at this junction, since
$\log2>31/45$. These costs are therefore below $N/u$.

It remains to compare the top height $v=\ell_u(a_0)$ with the
baseline. The profile derivative gives $u-v\le7h/4$.
In the quadratic cost, combine the degree corrections before
estimating them: its numerator is at most
\[
 \frac{\pi^2}{4}
 \left[1+\Delta_1a_0+\frac{\Delta_2}{2}(1-a_0)\right]-4h.
\]
This keeps the degree-one loss $1-k(1)$. Put $d=1-a_0=3h$.
Using $1-k(1)<1/\gamma$ and $k(1)-k(2)/2<1/\gamma$, we obtain
\[
 1-k(1)\rho^2a_0-k(2)\rho^4d/2
 \le(1+d)(\gamma^{-1}+4h^2)=4h(1+4h+3h^2).
\]
Also $u-v\le7h/4$ and $u\ge8$ imply
\[
 1-\frac uv(a_0+\rho^2d/2)\ge\frac54h,\qquad
 \frac{4u}{v}(1+4h+3h^2)<\frac{101}{25}.
\]
Writing $\overline{\mathcal Q}_u(a_0)$ for the upper cost and
using $u/v\ge1$ in the arcsine saving now gives
\begin{align*}
 \frac{N-u\overline{\mathcal Q}_u(a_0)}h
 &>3\left(\frac u2+\frac9{10}\right)
       -\frac{101}{10}-\frac37(2u+1)+4\\
 &=\frac{9u}{14}-\frac{134}{35}
 \ge\frac{46}{35}>0.
\end{align*}
The four terms record the spectral gain, degree-correction cost,
entropy penalty, and arcsine saving, respectively; their sum is a
positive linear function of $u$. All costs are therefore below $N/u$,
proving the derivative assertion by
Theorem~\ref{as:gate}.
\end{proof}

\par
\endgroup

\subsection{Completion of the proof}\label{ub:high-cover}

The compact spectral calculation uses the same moving local cutoff
throughout. It consequently gives a derivative criterion for every
possible positive-gap source without a separate neighborhood theorem.

\begin{proposition}\label{ub:compact-gate}
For every increasing $f:\{-1,1\}^n\to\{-1,1\}$, if
$49/50\le\rho\le\tanh8$ and $G_f(\rho)>0$, then
$G_f'(\rho)>0$.
\end{proposition}
\begin{proof}[Computer-assisted proof]
Lemma~\ref{ub:exponential-local} rules out positive gap when
$\alpha\ge1-2e^{-u}$. For the remaining sources,
Theorem~\ref{as:gate} gives $D(T_\rho f)>\rho u$, hence
$G_f'(\rho)>0$. The certificate verifies this criterion on
24 closed noise intervals covering $[45951/20000,8]$, always
with the cutoff $1-2e^{-u}$.
Since $\log9<45951/20000<\atanh(49/50)$, both the local cutoff
and the spectral cover apply throughout the stated range.
Appendix~\ref{app:certificate} gives the verification details.
\end{proof}

\begin{proof}[Proof of Theorem~\ref{jp:main}]
Proposition~\ref{ub:lower-interval} gives $0\le\rho\le457/500$;
Theorems~\ref{ub:upper-middle} and~\ref{ub:last-middle} give
$457/500\le\rho\le49/50$. Fix an increasing nonconstant source.
We claim that every positive-gap point with $49/50\le\rho<1$
has $G_f'(\rho)>0$. Proposition~\ref{ub:compact-gate} proves
this through $u=8$. For $u\ge8$, Lemma~\ref{rt:local-junction}
excludes positive gap when $\delta(f)\le(3/2)e^{-u}$; the dictator
case is exact. Proposition~\ref{qual:tail-derivative} covers every
remaining source. The claim follows.

Since $G_f(1)=-\Phi(\mu)\le0$, the differential barrier in
Lemma~\ref{ub:barrier} now proves CK throughout $[49/50,1]$.
Compress an arbitrary source at the fixed correlation in question;
constants have zero information. All intervals are closed, and the
endpoint assertions are $0\le0$ and $H(\mu)\le\log2$.
\end{proof}

\section*{Acknowledgements}

Vu Khac Ky thanks Professor Chandra Nair for introducing him to this problem and sharing valuable insights and ideas during his postdoctoral appointment under Professor Nair's supervision at the Chinese University of Hong Kong.

The authors formulated the approach and used ChatGPT to develop the proofs, prepare the verification code, and revise the exposition. The authors take full responsibility for the paper’s content.

\appendix

\section{Certificates and reproducibility}
\label{sec:certificates}\label{app:certificate}
Fixed low-correlation signs are checked by rational series bounds;
the other transcendental signs are checked on complete closed rational
domains by outward-rounded Arb arithmetic \cite{Johansson2017}. The moving-cap source
boundaries are quadratic algebraic numbers reconstructed exactly. A strict sign is accepted only when
the entire enclosure has that sign; a nonfinite enclosure or one
containing zero is inconclusive. Exact rational and algebraic arithmetic checks coverage.
Every supporting line is valid by convexity; its explicit integral
is accepted only after the required denominator and final signs pass.
The compact inverse profile is bounded analytically. Bisections retain both children unless infeasibility is
proved. Data-free scripts check the finite constants in the analytic
argument of Proposition~\ref{qual:tail-derivative} and the local
junction in Lemma~\ref{rt:local-junction}.

\subsection{The compact certificate}

The qualitative file \path{qualitative_compact_certificate.json}
covers $[45951/20000,8]$ by 24 closed rational noise intervals:
eight equal intervals through $4$, followed by 16 intervals of
length $1/4$. The cutoff is $1-2e^{-u}$ throughout, and
Lemma~\ref{ub:exponential-local} covers every larger influence.
The single explicit parameter rule uses 96 height endpoint inequalities
from Lemma~\ref{as:convex-majorant}. No height records or fitted
parameter families are stored. The split is $2u/3$, and
Lemma~\ref{as:explicit-inverse} supplies the upper endpoint without
root calculations. Exact adjacency checks verify the noise partition.

The replay \path{qualitative_compact_verify.py} recomputes every
scalar sign at 256 bits. It clears positive denominators and expands
logarithms of profile ratios before interval evaluation, retaining
cancellation without changing the signs. The transfer condition is
proved analytically. The separate \path{qualitative_compact_audit.py}
uses independently derived unscaled formulas, 384-bit arithmetic,
and first-order bounds on the full extended height domain.
Both verify the full stated domain and its moving influence cutoff. Neither uses an optimizer or inverse roots.

\subsection{Files and replay scope}

The certificate data and verification code are available in the
GitHub repository linked below. The low-correlation argument
uses ten fixed endpoint inequalities and no correlation partition.
The middle-range checks use ten scalar endpoint gates, five local
endpoint signs, and 29 clocks with 56 integration intervals and 152
supporting-line evaluations. The two-coordinate mean bound uses eight
fixed rational signs and no mean partition. The unbounded endpoint
argument is analytic; its checker verifies the finite constants used
in that argument, without sampling the unbounded domain.

From the bundle root, install \path{requirements.txt} and run
\texttt{python3 verify.py}. This checks the file manifest and runs
the ten numerical stages supporting the exact theorem. Each stage
recomputes its signs and domain coverage and writes its report to a
fresh output directory. Optimized Python is rejected because the
verifiers use assertions. The separate stability paper has its own
package and stronger quantitative checks; none of them is an input
to this exact proof. The \path{README.md} records the build and replay
commands. Only current inputs are included. The certificate data and verification code are available at
\url{https://github.com/vukhacky/courtade-kumar-conjecture/}.

\subsection{Growth of the scaled profile}

This estimate gives the explicit inverse bound in
Lemma~\ref{as:explicit-inverse} using only first derivatives.

\begin{lemma}\label{app:scaled-profile-derivative}
For the scaled profile $S$ in \eqref{eq:scaled-profile} and $v\ge2$,
$1<(2v+1)S'(v)/S(v)<2$.
\end{lemma}
\begin{proof}
Put $q=e^{-2v}$. The entropy formula gives
\[
 S(v)=2v+1+\sum_{n\ge1}(2v+1+c_n)q^n,\qquad
 c_n=\sum_{k=1}^n\frac{(-1)^{k+1}}{k(k+1)}\in(0,1/2].
\]
Termwise differentiation gives
$2-(4v+3)q/(1-q)^2\le S'(v)<2$.
The subtracted term decreases for $v\ge2$; $e^4>54$ bounds
it by $594/2809<1/4$, so $S'>7/4$.
Equation~\eqref{rt:scaled-profile} gives
$2v+1<S(v)\le(2v+1)(1+q)/(1-q)$.
Since $(1+q)/(1-q)<55/53<7/4$, the asserted ratios follow.
\end{proof}

\subsection{The central energy remainder}

The next estimate completes the affine interpolation argument in
Proposition~\ref{ub:lower-interval}. Its proof reduces the entire
noise interval to a concave polynomial and one endpoint value.

\begin{lemma}\label{app:central-energy}
For $0\le\rho\le457/500$,
\[
 \frac{76-49\rho^2}{57}H(\rho)
 <(1-\rho)\left(1+\rho-\frac{193027}{250000}\rho^2\right).
\]
\end{lemma}
\begin{proof}
Let $V(\rho)$ be the right side minus the left side. We show that
$V$ decreases and check its value at the right endpoint.
The entropy series gives
\[
 \frac{V'(\rho)}\rho
 =-2\left(1+\frac{193027}{250000}\right)+\frac43
  +\frac{98}{57}\log2+\frac{579081}{250000}\rho
  +\sum_{n\ge1}d_n\rho^{2n},
\]
where $d_n=4/[3(2n+1)]-49(n+1)/[57n(2n-1)]$.
We have
$d_1<-80/63$, $d_2<-17/105$, $d_3<-4/105$, and $d_4<0$.
For $n\ge5$, the bound $d_n<1/64$ follows from
\[
 84n^3-1280n^2+5227n+1152
 =(84n+64)(n-8)^2+875n-2944>0.
\]
Indeed, this polynomial is
$1344n(4n^2-1)$ times the difference between $1/64$ and the
upper bound obtained by replacing $49/57$ by $6/7$.
Also $\rho^{10}/[64(1-\rho^2)]<1/25$ on the stated interval.
Therefore $\log2<7/10$ gives $V'(\rho)/\rho<P(\rho)$, where
\[
 P(\rho)=-\frac{362581}{375000}+\frac{579081}{250000}\rho
          -\frac{80}{63}\rho^2-\frac{17}{105}\rho^4
          -\frac4{105}\rho^6.
\]
Now $P''<-5/2$, $P(3/4)<-19/10000$, and
$0<P'(3/4)<17/200$. Completing the square in its quadratic
tangent bound gives
$P(\rho)\le P(3/4)+P'(3/4)^2/5<-1/2500$.
Thus $V$ decreases. Finally $V(457/500)>1/50000$ by the
rational logarithm bounds below, proving the claim.
\end{proof}

For completeness, all ten low-correlation comparisons use explicit
finite series. Reduce a positive rational logarithm argument to
$x\in[1,2]$ by extracting powers of $2$. With $z=(x-1)/(x+1)$,
the first twelve terms of
$2\sum_{j\ge0}z^{2j+1}/(2j+1)$ underestimate $\log x$ by at
most $2z^{25}/[25(1-z^2)]$. The same formula bounds $\log2$.
For the Fourier constants, twelve terms of the alternating
arctangent series in
$\pi=16\arctan(1/5)-4\arctan(1/239)$ give rational bounds for
$\pi$; squaring rational endpoints bounds each square root.
The file \path{low_correlation_verify.py} performs these rational
operations with no input certificate. The separate
\path{low_correlation_audit.py} checks the same constants directly
with Arb at 384 bits.

\subsection{A rational curvature bound}

The following calculation supplies the curvature estimate used in
Lemma~\ref{ub:cubic-endpoints}. It involves only polynomial
coefficients and rational inequalities.

\begin{lemma}\label{app:cubic-curvature}
For $0\le x\le43/500$ and $3/4\le M\le83/100$, put
$C=3-2M+(-3+5M)x+(1-4M)x^2+Mx^3$ and
$V=(1/2+3x-3x^2+x^3)/C$.
Then $-7+30x<V''<0$ and $V'<2$.
\end{lemma}
\begin{proof}
All derivatives in this proof are with respect to $x$.
The denominator is positive: $C\ge67/50$, since
$C(0)=3-2M\ge67/50$ and
$C'\ge3/4-(116/25)(43/500)>0$.
Expanding the numerator of $-V''$ and bounding its coefficients over the
given interval for $M$ gives
$-C^3V''\ge15-60x-54x^3-7x^5>8$.
In particular, $V''<0$.
Also $V'(0)=(21/2-17M/2)/(3-2M)^2<2$, so $V'<2$.

For the lower bound, write
$C^3(V''+7-30x)=\sum_{j=0}^3P_j(x)(83/100-M)^j$.
Direct expansion, discarding positive terms and using $x<1/10$,
gives the following bounds:
\[
\begin{array}{c|l|c}
j&\text{lower bound for }P_j(x)&\text{further lower bound}\\ \hline
0&\frac54+(\frac{53}{2}-304x)x-588x^5&\frac65\\
1&57-316x-250x^2&22\\
2&165-1123x-2001x^3&50\\
3&56-660x+3186x^2-8742x^3&13
\end{array}
\]
The first row uses $53/2>304(43/500)$; the last uses
$56-660x+3186x^2\ge56-660^2/(4\cdot3186)$.
The two middle rows follow by setting $x=1/10$.
Thus every coefficient $P_j$ is positive, which proves the lower bound.
\end{proof}

\end{document}